\documentclass{article}
\usepackage[utf8]{inputenc}
\usepackage[margin=1in]{geometry}
\usepackage{style}
\usepackage{mdframed}
\usepackage{nicefrac}
\usepackage[colorlinks]{hyperref}
\usepackage[ruled,linesnumbered]{algorithm2e}
\usepackage{amsmath}
\usepackage{amssymb}
\usepackage{complexity}

\newcommand{\Var}{\text{Var}}
\newcommand{\Ex}{\E}

\renewcommand{\cP}{{\cal P}}

\allowdisplaybreaks

\usepackage{thm-restate}

\DeclareMathOperator{\Cost}{cost}

\usepackage{boxedminipage}

\usepackage[
natbib=true,
style=trad-alpha,
maxalphanames=8,
sorting=nty
]{biblatex}
\renewcommand{\cL}{\mathcal{L}}
\renewcommand{\cS}{\mathcal{S}}

\title{Constraint Satisfaction Problems with Advice}

\date{}
\author{
Suprovat Ghoshal\thanks{Supported by NSF Awards CCF-1955351.} \\ Northwestern and TTIC
\and
Konstantin Makarychev\thanks{Supported by NSF Awards CCF-1955351 and EECS-2216970.} 
\\ Northwestern  
\and
Yury Makarychev\thanks{Supported by NSF Awards CCF-1955173, CCF-1934843, and ECCS-2216899.} \\ TTIC 
}
\date{}
\begin{document}
\maketitle
\begin{abstract}
We initiate the study of algorithms for constraint satisfaction problems with ML oracle advice. We introduce two models of advice and then design approximation algorithms for Max Cut, Max 2-Lin, and Max 3-Lin in these models. In particular, we show the following. \\
\begin{itemize}
\item For Max-Cut and Max $2$-Lin, we design an algorithm that yields near-optimal solutions when the average degree is larger than a threshold degree, which only depends on the amount of advice and is independent of the instance size. We also give an algorithm for nearly satisfiable Max $3$-Lin instances with quantitatively similar guarantees.
\item Further, we provide impossibility results for algorithms in these models. In particular, under standard complexity assumptions, we show that Max $3$-Lin is still $1/2 + \eta$ hard to approximate given access to advice, when there are no assumptions on the instance degree distribution. Additionally, we also show that Max $4$-Lin is $1/2 + \eta$ hard to approximate even when the average degree of the instance is linear in the number of variables. 
\end{itemize}
\end{abstract}
\section{Introduction}
In recent years, numerous breakthroughs have occurred in machine learning (ML), and today, ML tools
can solve tasks that were completely out of reach even a decade ago.  This has sparked an interest in designing and using algorithms and data structures that rely on ML oracle advice (see e.g.~\cite{BDSV18,M18,PSK18,HIKV19, GP19, LV21}). 
Algorithms with access to ML advice can be used in the setting when we can learn some unknown information about the problem at hand using machine learning tools. This can be some information about the result or optimal solution (e.g., the position of an element in the sorted array~\cite{lu2021generalized,carvalho2023learnedsort,bai2024sorting} or membership in a set~\cite{kraska2018case,M18}), information about future events or yet unread data in the input stream (see e.g.,~\cite{mahdian2007allocating,devanur2009adwords,vee2010optimal,munoz2017revenue,GP19,lattanzi2020online,mitzenmacher2020scheduling}), as well as some information about the optimal set of parameters for the algorithm~\cite{BDSV18,DILMV21}.

In this paper, we introduce two models for solving constraint satisfaction problems (CSPs) with oracle advice. Suppose that we are given a constraint satisfaction problem with predicates $\Psi=\{\psi_1,\dots, \psi_m\}$ and Boolean variables $x_1,\dots, x_n$.\footnote{More generally, we can also consider the case where variables take values in some domain $[d]=\{1,\dots,d\}$.} It will be convenient for us to assume that value $1$ represents false and $-1$ represents true. Let $x^*= (x_1^*,\dots, x_n^*)$ be a fixed optimal solution, which we will refer to as the ground-truth solution. Now we assume that we are given noisy advice about $x_i^*$. Specifically, we consider two models.



\begin{itemize}
    \item {\bf Label Advice}. In the first model, the algorithm receives advice $\tilde x= (\tilde x_1,\dots, \tilde x_n)$, where $\tilde x_i$ is a noisy prediction of the ground-truth value $x_i^*$  Each $\tilde x_i$ is a random variable taking values $-1$ and $1$ and is slightly biased toward $x_i^*$. Namely, $\tilde x_i = x_i^*$ with probability $\frac{1+\varepsilon}{2}$ and $\tilde x_i = - x_i^*$ with probability $\frac{1-\varepsilon}{2}$. All variables $\tilde x_i$ are independent.

    \item {\bf Variable Subset Advice}. In the second model, the algorithm receives a random subset of variables/indices  $S\subseteq \{1,\dots, n\}$ and their values $(x^*_i)_{i\in S}$. Subset $S$ includes every $i$ with probability $\varepsilon$; all events $i\in S$ are independent.
 \end{itemize}

These models capture the setting where we have an ML algorithm (or another oracle) that provides unreliable predictions for values $x_i^*$.  Since the predictions are very noisy, they do not provide a good solution. Consider, for example, the Label Advice model for Max 2-Lin, a constraint satisfaction problem with constraints $x_i\cdot x_j = -1$ and $x_i\cdot x_j = 1$. Even if the optimal solution satisfies all the constraints,
solution $\tilde x$ satisfies only a $\frac{1+\varepsilon^2}{2}$ fraction of the constraints in expectation; this is just a tiny bit better that the $1/2$ fraction that a random solution satisfies.
The aim of this paper is to show that, nevertheless, advice $\tilde x$ may be very valuable. In this paper, we focus on the Max $k$-Lin problem. Our choice of the
Max $k$-Lin problem is not arbitrary. H{\aa}stad~\cite{Has01} proved that Max $3$-Lin is approximation resistant (as well as every Max $k$-Lin problem with $k\geq 3$). In other words, it is not possible to obtain a solution satisfying $\nicefrac{1}{2}+\eta$ fraction of all constraints in polynomial time even if the optimal solution satisfies $1-\delta$ fraction of all constraints (for every positive $\delta$ and $\eta$; assuming $P\neq NP$). For Max 2-Lin and Max Cut problems, the best approximation can be obtained using the Goemans-Williamson algorithm~\cite{GW95,KKMO07}.
We show how to get nearly optimal solutions using oracle advice for Max Cut, Max 2-Lin and almost satisfiable instances of Max 3-Lin when the number of constraints is sufficiently large (for unweighted instances, the number of constraints should be at least $C_{\varepsilon} n$ for Max Cut and $C_{\varepsilon,\delta} n$ for Max 3-Lin, where $C_{\varepsilon} n$ and $C_{\varepsilon,\delta} n$ depend only on $\varepsilon$ and $\varepsilon,\delta$, respectively). We complement our algorithmic results with hardness of approximation bounds. Specifically, we show that there are no polynomial-time algorithms for Max $k$-Lin instances with $k\geq 3$ in the Label Advice and Variable Subset Advice models  that given a $(1-\delta)$-satisfiable instance and advice find a solution satisfying at least a $(1/2 + \delta)$-fraction of the constraints with constant probability (for every $\delta > 0$). Furthermore, for $k\geq 4$, we show that there are no algorithms for $(1-\delta)$-satisfiable Max $k$-Lin instances with $n$ variables and $\Theta(n^2)$ constraints (``high-degree instances'') that satisfy at least a $0.99$ fraction of the constraints with high probability.

\paragraph{Comparing the Models.} We note that the Variable Subset Advice model provides more information than the Label Advice one. Indeed, given set $S$ and values $(x_i^*)_{i\in S}$, we can generate advice $\tilde x$ as follows: if $i\in S$, let $\tilde x_i = x_i^*$; otherwise, sample $\tilde x_i$ uniformly at random. It is immediate that 
$$\Pr(\tilde x_i = x_i^*) = \Pr(x_i = x_i^* \mid i\in S)\Pr(i\in S) + 
\Pr(x_i = x_i^* \mid i\notin S)\Pr(i\notin S) = 1\cdot \varepsilon+\nicefrac{1}{2}\cdot (1-\varepsilon) = \frac{1+\varepsilon}{2}$$
as required, and all $\tilde x_i$ are independent. Thus, every algorithm for the Label Advice model also works in the Variable Subset Advice model. For this reason, we will consider the Label Advice model in this paper.

\paragraph{Max Cut and Max $k$-Lin Problems.}
We recall the definitions of Max Cut and Max 2-Lin problems.
\begin{definition}[Max Cut]
In Max Cut, we are given an undirected graph $G=(V,E)$ with edge weights $w_e>0$. The goal is to find a cut $(S,T)$ that maximizes the total weight of cut edges.

Alternatively, Max Cut can be stated as a constraint satisfaction problem. We are given a set of Boolean variables $x_1,\dots, x_n$ and a set of constraints of the form $x_i \cdot x_j = -1$ (or, equivalently, $x_i\neq x_j$); each constraint has a non-negative weight. The goal is to find an assignment that maximizes the total weight of satisfied constraints.
\end{definition}
The connection between the graph and CSP formulations of Max Cut is straightforward: vertex $v_i$ corresponds to variable $x_i$ and edge $(v_i, v_j)$ corresponds to constraint $x_i\cdot x_j =-1$. If $v_i \in S$ then $x_i = -1$; if $v_i \in T$ then $x_i = 1$.
We now define Max 2-Lin, which is a generalization of Max Cut.
\begin{definition}[Max 2-Lin]
In Max 2-Lin, we are given a set of Boolean variables $x_1,\dots, x_n$ and a set of constraints of the form $x_i \cdot x_j = c_{ij}$ where $c_{ij}\in \{-1, 1\}$; each constraint has a non-negative weight $w_{ij}$. The goal is to find an assignment that maximizes the total weight of satisfied constraints.
\end{definition}

Both problems Max Cut and Max 2-Lin are NP-hard~\cite{Karp}. The Goemans--Williamson algorithm provides an $\alpha_{GW}=0.878\dots$ approximation~\cite{GW95} for them and, as Khot, Kindler, Mossel, and O’Donnell showed, this is optimal assuming the Unique Games Conjecture~\cite{KKMO07}. If a Max Cut or Max 2-Lin instance is almost satisfiable -- that is, the value of the optimal solution is $(1-\delta) W$, where $W$ is the total weight of all the constraints/edges -- then the Goemans--Williamson algorithm finds a solution of value $1- O(\sqrt{\delta})$. This result is again optimal assuming the Unique Games Conjecture~\cite{KKMO07}. See \cite{MM17} for a detailed discussion of these and other approximation and hardness results for constraint satisfaction problems.

\begin{definition}[Max $k$-Lin]
In Max $k$-Lin, we are given a set of Boolean variables $x_1,\dots, x_n$ and a set of constraints of the form $x_{i_1} \cdots x_{i_k} = c_{i_1\dots i_k}$ where $c_{i_1\dots i_k}\in \{-1, 1\}$; each constraint has a non-negative weight $w_{ij}$. The goal is to find an assignment that maximizes the total weight of satisfied constraints.
\end{definition}

Unlike the case of Max $2$-Lin, one cannot find near-optimal solutions to nearly satisfiable instances of Max $3$-Lin. In particular, H\r{a}stad~\cite{Has01} showed that for Max $4$-Lin is approximation resistant i.e., for instances where the optimal solution satisfies at least $1 - \delta$ fraction of constraints, it is NP-hard to find an assignment that satisfies more than $1/2 + \delta$ fraction of constraints, for every small constant $\delta > 0$.

\subsection{Our results}
We design polynomial-time approximation algorithms for Max Cut, Max 2-Lin, and Max 3-Lin. In unweighted graphs, the algorithm for Max Cut finds a nearly optimal solution if the average degree $\Delta = 2m/n$ is a sufficiently large constant (for a fixed parameter $\varepsilon$), specifically $\Delta \geq C/\varepsilon^2$ (here $m$ is the number of constraints/edges, $n$ is the number of variables/vertices; $C$ is a constant). 

\begin{theorem}\label{thm:Max2Lin-unweighted}
There exists a polynomial-time algorithm for the Label Advice model that given an unweighted instance of Max Cut or Max 2-Lin and advice $\tilde x$ finds a solution of value at least $(1-O(\nicefrac{1}{\varepsilon\sqrt{\Delta}})) \, OPT$ in expectation (over the random advice), where $OPT$ is the value of the optimal solution, $\varepsilon$ is the parameter of the model, and $\Delta$ is the average degree (see above).
\end{theorem}

In weighted graphs, our results require that $n\sum_{ij}\left(\frac{w_{ij}}{W}\right)^2\ll \varepsilon^2$, where $W$ is the total weight of all constraints.
\begin{theorem}\label{thm:Max2Lin}
There exists a polynomial-time algorithm for the Label Advice model that given an instance of Max Cut or Max 2-Lin and advice $\tilde x$ finds a solution of value at least $OPT - \sqrt{n\sum_{ij}w_{ij}^2}/\varepsilon$ in expectation (over random advice), where $OPT$ is the value of the optimal solution, $\varepsilon$ is the parameter of the model, and $w_{ij}$ is the weight of the constraint for $x_i$ and $x_j$.
\end{theorem}

Next, we state our result for the Max $3$-Lin problem. As stated in the theorem below, we show that for nearly satisfiable Max $3$-Lin instances, there exists an efficient algorithm that finds a nearly satisfying assignment, if the average degree of the instance is large enough. 

\begin{restatable}{rethm}{alg}				\label{thm:3LIN}
	 There exists a polynomial-time algorithm that given a $(1-\delta)$-satisfiable (unweighted) instance~$\Phi$ of Max 3-Lin having $n$ variables and at least $\delta^{-1}\ln \frac{1}{\delta} \cdot \varepsilon^{-6}\cdot n$ constraints along with a label advice with parameter $\varepsilon>0$, returns a solution that satisfies a $1-O(\sqrt{\delta})$ fraction of all the constraints.
\end{restatable}

\paragraph{Hardness Results.} We complement the above positive results with a couple of hardness results. Our first result is the following theorem which shows that under standard assumptions, Max $3$-Lin retains its $(1/2 + o(1))$-inapproximability if there are no guarantees on its average degree, even in the setting where the algorithm has access to the oracle advice.

\begin{restatable}{rethm}{hardness}				\label{thm:3lin-hardness}
    Assume that the Exponential Time Hypothesis (ETH) and Linear Size PCP Conjecture hold. For every $\delta > 0$, there exists $\varepsilon_0 = \varepsilon_0(\delta)$ such that for every $\varepsilon\in [0, \varepsilon_0]$, there is no polynomial time algorithm for Max $3$-Lin in the variable subset advice model with parameter $\epsilon$ that given a $(1-\delta)$-satisfiable instance returns a solution satisfying at least a $(1/2+\delta)$-fraction of the constraints with probability at least 0.9 over the random advice.
\end{restatable}

Next, we show that for the Max $4$-Lin problem, there are no efficient algorithms in the oracle advice setting that can find nearly optimal solutions, even when the average degree of the instance is large.

\begin{restatable}{rethm}{ldeg}				\label{thm:4lin-hardness}
    Assume that the Exponential Time Hypothesis (ETH) and Linear Size PCP Conjecture hold. Then there exists a constant $\delta_0 \in (0,1)$ such that for every $\delta \in (0,\delta_0)$, there exists $d(\delta) \in \mathbbm{N}$ and $\epsilon(\delta) \in (0,1)$ for which the following statement holds. For every $\epsilon \in (0,\epsilon(\delta))$, $n \geq n(\delta)$, and $d \in [d(\delta),n]$, there is no polynomial time algorithm for Max $4$-Lin in the variable subset advice model with parameter $\epsilon$ that, given a $(1 - \delta)$-satisfiable Max $4$-Lin instance with $n$ variables and at least $\Omega(d n)$ constraints, returns a solution satisfying a least a $(1/2 + \delta)$-fraction of the constraints with probability at least $0.9$ over the random advice.
\end{restatable}

\begin{remark} Theorem~\ref{thm:3lin-hardness} holds not only for Max 3-Lin but also for Max $k$-Lin with $k\geq 3$ and Theorem~\ref{thm:4lin-hardness} holds for Max $k$-Lin with $k\geq 4$. To see the former, note that any instance of 3-Lin can be converted to an equivalent instance of Max $k$-Lin as follows. Let us introduce $k-3$ new variables $y_1,\dots, y_{k-3}$ and then replace each constraint $x_ax_bx_c =c_{abc}$ with $x_ax_bx_cy_1\cdots y_{k-3} =c_{abc}$. Every solution of the original Max 3-Lin instance can be extended to a solution of the obtained Max $k$-Lin by letting $y_1=\dots=y_{k-3}=1$ and keeping all $x_i$ as is. Also, every solution of Max $k$-Lin can be transformed to a solution of Max $3$-Lin: if $y_1\cdots y_{k-3} = 1$, keep all $x_i$ as is; otherwise, negate all $x_i$. The same transformation shows that Theorem~\ref{thm:4lin-hardness} applies when $k\geq 4$.
In this case, however, the analysis of the transformation is non-black-box but still pretty simple.
\end{remark}

Note that while Theorems \ref{thm:3lin-hardness} and \ref{thm:4lin-hardness} have been stated for the variable subset advice model, by the discussion above, it follows that they also apply to the label advice model as well (with parameter $\epsilon/2$ instead of $\epsilon$).

\paragraph{Discussion.} Our results for Max-$2$-Lin and Max-Cut (Theorem \ref{thm:Max2Lin-unweighted}) show that using oracle advice we can find a near-optimal solution as long as the average degree of the instance is large enough, irrespective of the optimal value of the instance (importantly, the degree requirement depends on the amount of oracle advice $\varepsilon$ but not on the size of the instance). Theorem~\ref{thm:Max2Lin} gives an analogous result for the weighted variants of these problems. For Max $3$-Lin, we obtain a similar guarantee for {\em nearly satisfiable} instances (with a different requirement on the average degree), and it is an interesting open question if one can find near-optimal solutions for Max $3$-Lin instances that are not nearly satisfiable, given access to oracle advice. To complement our positive results, Theorem \ref{thm:3lin-hardness} shows that if there are no assumptions on the average degree of the instance, Max $3$-Lin is still $1/2 + o(1)$ hard to approximate even with oracle advice (under certain complexity assumptions). That is, our assumption on the average degree (or another related assumption) is necessary in our algorithmic result for Max 3-Lin. Finally, Theorem~\ref{thm:4lin-hardness} shows that Max 4-Lin is much more difficult than Max 2-Lin and Max 3-Lin: even if the input instance is nearly satisfiable and has a high average degree (the number of constraints is quadratic in the number of variables), it is hard to find a solution that satisfies strictly better than $1/2$-fraction of the constraints (in our oracle-advice models).

\subsection{Comparison with
Independent Works}
Independently and concurrently with our work, Cohen-Addad, d'Orsi, Gupta, Lee, Panigrahi~\cite{CdGLP24} proposed the same model of ML oracle advice for constraint satisfaction problems. In their work,
they provide approximation algorithms for Max Cut that have approximation factors $\alpha_{GW}+\tilde\Omega(\varepsilon^4)$ and $\alpha_{GW}+\Omega(\varepsilon)$ 
in the Label Advice and Variable Subset Advice models, respectively (they call these models \emph{Noisy Predictions} and \emph{Partial Predictions}). They also present an algorithm that finds almost optimal solutions for Max Cut on 
$\Delta$-wide graphs, a notion introduced in their paper. Loosely speaking, 
a graph is $\Delta$-wide if not all weight is concentrated on few edges. We refer the reader to~\cite{CdGLP24} for details. Given an instance of Max Cut satisfying $OPT$ constraints, the algorithm by Cohen-Addad et al.~\cite{CdGLP24} finds a solution of value $(1-O(\eta +\nicefrac{1}{\varepsilon\sqrt{\Delta}}))\,OPT$.
This result is somewhat similar but not directly comparable to our Theorems~\ref{thm:Max2Lin-unweighted} 
and \ref{thm:Max2Lin}.
For example, in regular unweighted and weighted $\Delta$-wide graphs, our algorithm finds a solution of value at least $(1-O(\nicefrac{\sqrt{\eta}}{\varepsilon\sqrt{\Delta}})) \, OPT$; this guarantee is stronger than the one given by Cohen-Addad et al.
However, our algorithm cannot be used on some $\Delta$-wide graphs, whose vertex degrees are very irregular. The authors complement their results with an algorithm for $\Delta$-narrow graphs i.e., graphs that are not $\Delta$-wide. This algorithm does not use any oracle advice. It is based on the algorithms by Feige, Karpinski,
and Langberg~\cite{FKL02} and Hsieh and Kothari~\cite{HK22} for low-degree instances of Max Cut. In their paper, Cohen-Addad et al.~\cite{CdGLP24} do not study Max $k$-Lin for $k>2$.

In an independent work~\cite{BEX24}, Bampis, Escoffier,  and Xefteris considered a similar (but not identical) notion of advice and provided PTAS for $\delta$-dense instances of Max Cut and other problems (a $\delta$-dense instance of Max Cut contains 
$\delta n(n-1)/2$ edges).
\section{Algorithm for Max Cut and Max 2-Lin with Advice}

Our algorithm finds a solution for Max 2-Lin instance by (approximately) maximizing the quadratic form 
\begin{equation}\label{eq:max-cut-qp}
\sum_{i,j} \frac{|a_{ij}| + a_{ij}x_ix_j}{4},
\end{equation}
where the matrix $A=(a_{ij})$ is defined as follows. For Max Cut, $A$ is minus adjacency matrix i.e., $a_{ij} = - w_{ij}$, where $w_{ij}$ is the weight of edge $(i,j)$. For Max 2-Lin, $a_{ij} = w_{ij}$ if we have constraint $x_ix_j = 1$ and $a_{ij} = -w_{ij}$ if we have 
$x_ix_j = -1$, $w_{ij}$ is the weight of the constraint for $x_i$ and $x_j$. Note that matrix $A$ is an $n\times n$ symmetric matrix with zero diagonal. We remark that  Quadratic Program~(\ref{eq:max-cut-qp}) was used by Goemans and Williamson~\cite{GW95} in their seminal paper on semi-definite programming algorithm for Max Cut.

In this section, we give an algorithm for Max Cut that finds a solution of value $(1-O(1/\sqrt{\Delta}\cdot \varepsilon^{-1})) \, OPT$, where $OPT$ is the value of the optimal solution, $\varepsilon$ is the parameter of our model, and $\Delta = 2m/n$ is the average vertex degree in the graph. Our algorithms give a nontrivial approximation when $\Delta \geq C/\varepsilon^2$.
Note that matrix $A=(a_{ij})$ is symmetric with a zero diagonal.

The value of quadratic form~(\ref{eq:max-cut-qp}) exactly
equals the number of satisfied constraints in the corresponding 
instance of Max 2-Lin. Indeed, if the constraint for $x_i$ and $x_j$ is satisfied, then the term
$|a_{ij}| + a_{ij}y_iy_j = 
2|a_{ij}|$; if it is not satisfied, then $|a_{ij}| + a_{ij}y_iy_j = 0$. Consequently,
$$\frac{|a_{ij}| + a_{ij}x_ix_j + 
|a_{ji}| + a_{ji}x_ix_j}{4}
= w_{ij},$$
if the constraint for $x_i$ and $x_j$ is satisfied; and 
$$\frac{|a_{ij}| + a_{ij}x_ix_j + 
|a_{ji}| + a_{ji}x_ix_j}{4}
= 0,$$
if the constraint is not satisfied. Since the sum of all $|a_{ij}|$ equals $2W$, where $W$ is the total weight of all constraints, quadratic form~(\ref{eq:max-cut-qp}) equals
\begin{equation}\label{eq:max-cut-qp-prime}
\frac{W}{2} + \sum_{i,j} \frac{a_{ij}x_ix_j}{4}.
\end{equation}

In the next section, we show how to obtain a solution of value $OPT - \varepsilon^{-1}\sqrt{n}\|A\|_F$ for Max QP (see Theorem~\ref{thm:q-max}). Here, $\|A\|_F$ is the Frobenius norm of $A$. This result implies Theorem~\ref{thm:Max2Lin} and also 
Theorem~\ref{thm:Max2Lin-unweighted}, because $OPT \geq W/2$.

\section{Quadratic Forms with Advice} 
In this section, we consider the quadratic form maximization problem with advice. Let $A=(a_{ij})$ be a symmetric $n\times n$ matrix. Our goal is to maximize the quadratic
form $\sum_{ij} a_{ij}x_ix_j$ for $x\in\{-1,1\}^n$. This quadratic form can also be written as 
$\langle x, Ax\rangle$.
We assume that the ground truth solution is $x^*$. The algorithm receives advice $\tilde{x}$. Each $\tilde{x}_i$ is a random variable, $\tilde{x}_i = x^*_i$ with probability $(1+\varepsilon)/2$ and $\tilde{x}_i = -x^*_i$ with probability $(1-\varepsilon)/2$. All variables $\tilde{x}_i$ are independent. The main result of this section is the following theorem.

\begin{theorem}				\label{thm:q-max}
	Let $A \in \mathbbm{R}^{n \times n}$ be a symmetric matrix with zero diagonal entries. Then, there exists a deterministic polynomial time algorithm that finds a labeling $x'$ such that randomizing over the choice of advice $\tilde{x}$, we have that 
\begin{equation}\label{eq:thm:q-max}
\Ex_{\tilde{x}}\left[\langle x', A x' \rangle\right] \geq \langle x^*, A x^* \rangle - \epsilon^{-1}\sqrt{n}\|A\|_F.
\end{equation}
\end{theorem}

The algorithm for the theorem is described below.

\begin{figure}[ht!]
\begin{mdframed}
\noindent{\bf Algorithm}

\medskip

\noindent {\bf Input:} Coefficient matrix $A$, oracle advice $\tilde{x}_1,\ldots,\tilde{x}_n \in \{\pm 1\}$.

\noindent {\bf Output:} Solution $x'_1,\dots,x'_n$.

\medskip

\begin{enumerate}
    \item Define $F(x,y) := \langle x, Ay \rangle - \|A(\epsilon x - y)\|_1$.
    \item Solve the following mathematical program with a concave objective:
    \begin{eqnarray}\label{eq:concave-MP}
        \textbf{Maximize} & F(x,\tilde{x})& \\
        \textbf{Subject to}& x_i \in [-1,1] & \forall~i \in [n]
    \end{eqnarray}
    \item Round the fractional solution $x$ coordinate-by-coordinate to a solution $x'\in\{\pm 1\}^n$ such that $\langle x', A x' \rangle \geq \langle x, A x \rangle$.
    \item Output labeling $x'$.
\end{enumerate}
\end{mdframed}
\caption{Quadratic Program Maximization Algorithm}
\label{fig:quad-max}
\end{figure}

\begin{proof}
First, observe that function $F(x,y)$ is a concave function of $x$ for a fixed $y$, since $\langle x, Ay \rangle$ is a linear function of $x$, and $\|A(\epsilon x - y)\|_1$ is a convex function (because all vector norms are convex). Thus, we can find the maximum of $F(x,\tilde{x})$ subject to the constraint $x\in[-1,1]^n$ in polynomial time. Also, note that $\langle x, Ax\rangle$ is a linear function of each $x_i$ when all other coordinates $x_j$ ($j\neq i$) are fixed. Thus, the algorithm can round each $x$ to a $x'\in\{\pm 1\}$ by rounding coordinates one-by-one. At every step, the algorithm replaces one coordinate $x_i\in[-1,1]$ with $-1$ or $+1$ making sure that the quadratic form $\langle x, Ax\rangle$ does not decrease. 

We now show that inequality~(\ref{eq:thm:q-max}) holds. We begin with the following claim.

\begin{claim}
For every $x,y \in [-1,1]^n$, we have
	\[
	\langle Ax, x \rangle \geq \frac{F(x,y)}{\epsilon}.
	\]
\end{claim}
\begin{proof}
Write:
\[
	\langle x, Ax \rangle = \frac{\langle x, A y \rangle + \langle x , A(\epsilon x - y) \rangle}{\epsilon} \geq \frac{\langle x, A y \rangle - \|A (\epsilon x - y)\|_1}{\epsilon} = \frac{F(x,y)}{\epsilon},
\]
where the inequality follows from H\"older's inequality:
$$|\langle x , A(\epsilon x - y) \rangle | \leq
\|A(\epsilon x - y)\|_1.
$$
\end{proof}

Next, we bound the expected value of the optimization program~(\ref{eq:concave-MP})

\begin{lemma}
Let $x \in [-1,1]^n$ denote the optimal solution of concave program~(\ref{eq:concave-MP}). Then,
\[	\Ex_{\tilde{x}}\left[F(x,\tilde{x})\right] \geq \epsilon\langle x^*, A x^* \rangle - \sqrt{n}\,\|A\|_F. 
\]
\end{lemma}
\begin{proof}
The ground truth solution $x^*$ is always a feasible solution to program~(\ref{eq:concave-MP}). Hence, we have
$$
\Ex_{\tilde{x}}\left[F(x,\tilde{x})\right]	\geq \Ex_{\tilde{x}}\big[F(x^*,\tilde{x})\big]
= \Ex_{\tilde{x}}\big[\langle x^*, A\tilde{x} \rangle - \|A(\epsilon x^* - \tilde{x})\|_1\big] 
= \epsilon \langle x^*, Ax^* \rangle - \Ex_{\tilde{x}}\big[\|A(\epsilon x^* - \tilde{x})\|_1\big]. 
$$
Here, we used that $\Ex \tilde{x} = \varepsilon x^*$. 
To bound the last term, we let $z = \epsilon x^* - \tilde{x}$ and observe that 
$\|Az\|_1 \leq \sqrt{n}\,\|Az\|_2$ by the 
Cauchy–Schwarz inequality (since $Az$ is an $n$-dimensional vector). Then,
$$
\Ex\|A(\epsilon x^* - \tilde{x})\|_1 =
\Ex\|Az\|_1
\leq \sqrt{n}\,\Ex\|Az\|_2.$$
By Jensen's inequality,
$$
\sqrt{n}\,\Ex\|Az\|_2
\leq \sqrt{n}\,\Ex\Big[\|Az\|^2_2\Big]^{1/2} = 
\sqrt{n}\,\Ex\Big[
\langle Az, Az
\rangle\Big]^{1/2}
=
\sqrt{n}\,\Ex\Big[
\langle z, A^*Az
\rangle\Big]^{1/2} = 
\sqrt{n}\,\Ex\Big[\sum_{ij}
(A^*A)_{ij} z_iz_j\Big]^{1/2}.
$$
Random variables $z_i$ are 
mutually independent and $\Ex[z_i]=0$ for all $i$. Thus,
$\Ex[z_i z_j] = 0$ if $i\neq j$.
Also, $\Ex[z_i^2] = \Var[\tilde{x}_i] = 1-\varepsilon^2$. Therefore,
$$\Ex\Big[\sum_{ij}
(A^*A)_{ij} z_iz_j\Big] =
(1-\varepsilon^2)\sum_i 
(A^*A)_{ii} =  (1-\varepsilon^2) \,\text{tr}(A^*A) = (1-\varepsilon^2)\|A\|^2_F.
$$
Putting the bounds together completes the proof.
\end{proof}

We are now ready to finish the proof of Theorem \ref{thm:q-max}. Recall that $x' \in \{-1,1\}^n$ is the integral solution obtained by greedy coordinate-wise rounding of $x$. Thus,
\[
\Ex_{\tilde{x}}\left[\langle x', Ax' \rangle \right]
\geq \Ex_{\tilde{x}}\left[\langle x, A x \rangle\right] 
\geq \Ex_{\tilde{x}}\left[\frac{F(x,\tilde{x})}{\epsilon}\right] \geq \frac{\epsilon \langle x^*, A x^* \rangle - \sqrt{n}\|A\|_F}{\epsilon} 
= \langle x^*, A x^* \rangle - \sqrt{n}\epsilon^{-1}\|A\|_F.
\]

\end{proof}

\section{Algorithm for Max 3-Lin with Advice}


\alg*

\noindent\textbf{Remark:} Instance $\Phi$ must have at least $\ln \frac{1}{\delta} \cdot \delta^{-1} \varepsilon^{-6}\cdot n$ constraints. This means that every variable should participate in $3\delta^{-1}\ln \frac{1}{\delta} \cdot \varepsilon^{-6}$ constraints on average.
\begin{proof}
Our algorithm works as follows. First, it creates a new \emph{weighted} instance $\Psi$ of Max 2-Lin. This instance has the same set of variables, $x_1,\dots,x_n$, as Max 3-Lin instance $\Phi$ but a different set of constraints, which are created by functions \textsc{Create-H-Constraints} and \textsc{Create-L-Constraints} (see below). Then, the algorithm solves the Max 2-Lin instance using the Goemans and Williamson algorithm  for MAX CUT~\cite{GW95} and obtains a solution $\hat{x}_1,\dots,\hat{x}_n$. Finally, it outputs 
$\hat{x}_1,\dots,\hat{x}_n$ as a solution to the original  Max 3-Lin instance. We now describe functions \textsc{Create-H-Constraints} and \textsc{Create-L-Constraints}.

Fix a threshold $t = 8\varepsilon^{-2}\ln \nicefrac{1}{\delta}$. Denote by $E(\Phi)$ the set of indices of all constraints in $\Phi$ i.e., for every constraint $x_i x_j x_j=c_{ijk}$, set $E(\Phi)$ contains an \emph{unordered} tuple $(i,j,k)$.
Let $E_{ij}=\{(i,j,k) \in E(\Phi)\}$ be the set of indices of all constraints in $\Phi$ that contain variables $x_i$ and $x_j$. We say that set $E_{ij}$ is heavy if it contains at least $t$ elements. If a constraint belongs to at least one heavy set, we call it heavy. In other words, constraint 
$x_ix_jx_k = c_{ijk}$ is heavy if 
$E_{ij}$, $E_{jk}$, or $E_{ik}$ is heavy. If a constraint is not heavy, we call it light.

For every \emph{heavy} set $E_{ij}$, 
function \textsc{Create-H-Constraints} creates $2|E_{ij}|$ constraints in $\Psi$. It first computes $$\sigma_{ij} = 
\operatorname{sgn}\Big(\sum_{k:(i,j,k)\in E_{ij}}c_{ijk}\tilde{x}_k\Big)$$
and then creates constraints $x_ix_j = \sigma_{ij}$ and
$x_k = \sigma_{ij}c_{ijk}$ in $\Psi$
for each constraint $x_ix_jx_k = c_{ijk}$ in $E_{ij}$. We say that these constraints are representatives for
$x_ix_jx_k = c_{ijk}$ in the new instance $\Psi$. 
Note that the first constraint ($x_ix_j = \sigma_{ij}$) is identical for all constraints in 
$E_{ij}$ and does not depend on $k$. However, the second constraint ($x_k = \sigma_{ij}c_{ijk}$) does depend on~$k$.  Strictly speaking, if $\sigma_{ij}=0$, constraint $x_i x_j = \sigma_{ij}$ is not  a valid constraint because the right hand side must be either $1$ or $-1$. To make this constraint valid, we replace $\sigma_{ij}$ with $1$ but nevertheless, we conservatively assume that this constraint is always violated.

For every variable $x_i$, function \textsc{Create-L-Constraints} finds all light constraints that contain $x_i$. Denote this set by $L_i$. Then, for each constraint 
$x_ix_jx_k = c_{ijk}$ in $L_i$, \textsc{Create-L-Constraints} creates a constraint 
$x_i=\sigma_i$, where
$$\sigma_{i} = 
\operatorname{sgn}\Big(\sum_{j,k:(i,j,k)\in L_{i}}c_{ijk}\tilde{x}_j\tilde{x}_k\Big).$$
We say that this constraint is a representative for $x_ix_jx_k = c_{ijk}$ in the new instance $\Psi$. Note that 
\textsc{Create-L-Constraints} creates identical constraints for all constraints in $L_i$. If $\sigma_i = 0$, then, as before, we replace $\sigma_i$ with $1$ but will assume that $x_i=\sigma_i$ is violated for every solution $x$.

\medskip

\noindent\textbf{Analysis.} We have completely described the algorithm and now proceed to the analysis. It is clear that our algorithm runs in polynomial time. We prove that the expected fraction of satisfied constraints is $1-O(\sqrt{\delta})$.
We first state the main technical lemma.

\begin{lemma}\label{lem:soln-Psi}
The ground truth solution $x^*_1,\dots,x^*_n$ to $\Phi$ satisfies a $(1-O(\delta))$ fraction of all the constraints in $\Psi$, in expectation.
\end{lemma}
We prove Lemma~\ref{lem:soln-Psi} in Section~\ref{sec:lem:soln-Psi}. Now, we show that Lemma~\ref{lem:soln-Psi} implies Theorem~\ref{thm:3LIN}.
First, observe that for every constraint in the original Max 3-Lin instance $\Phi$ our algorithm creates 2, 3, 4, or 6 constraints in Max 2-Lin instance $\Psi$: For each light constraint, it creates exactly 3 constraints; for each heavy constraint, it creates from 2 to 6 new constraints. Also, note that every constraint in $\Psi$ is a representative for exactly one constraint in $\Phi$.
By Lemma~\ref{lem:soln-Psi}, the optimal solution satisfies a $(1-O(\delta))$ fraction of the constraints in $\Psi$, in expectation. Thus, the Goemans--Williamson algorithm~\cite{GW95} finds a solution satisfying a $(1-O(\sqrt{\delta}))$ fraction of the constraints in $\Psi$, in expectation (see also Theorem 2 in survey~\cite{MM17})\footnote{The intermediate Max 2-Lin instance $\Psi$ may contain duplicate constraints, and hence the analysis counts the constraints with their multiplicities. However, this is fine since the result of \cite{GW95} holds for weighted instances as well.}. This means that at most $O(\sqrt{\delta}\,m)$ constraints in Max 2-Lin instance $\Psi$ are not satisfied by 
$\hat{x}_1,\dots,\hat{x}_n$, in expectation.

We now bound the number of constraints in $\Phi$ not satisfied by $\hat{x}_1,\dots,\hat{x}_n$. We separately consider heavy and light constraints. 
If $x_i x_j x_k = c_{ijk}$ is a heavy constraint, then at least one of the sets $E_{ij}$, $E_{jk}$, $E_{ik}$ must be heavy. Assume without loss of generality that $E_{ij}$ is heavy. Then, 
$x_i x_j x_k = c_{ijk}$ must have two representatives in $\Psi$: 
$x_i x_j = \sigma_{ij}$ and 
$x_k = \sigma_{ij} c_{ijk}$. If both of these constraints are satisfied in $\Psi$ by 
$\hat{x}_1,\dots,\hat{x}_n$, then 
$x_i x_j x_k = c_{ijk}$ is also satisfied, because
$$\hat{x}_i \hat{x}_j \hat{x}_k = (\hat{x}_i \hat{x}_j) \hat{x}_k = 
\sigma_{ij}(\sigma_{ij} c_{ijk}) = (\sigma_{ij})^2 c_{ijk}= c_{ijk}.$$
Therefore, if a heavy constraint in $\Phi$ is not satisfied by 
$\hat{x}_1,\dots,\hat{x}_n$, then one of its representatives is not satisfied by $\hat{x}_1,\dots,\hat{x}_n$. Since every constraint in $\Psi$ is a representative for exactly one constraint in $\Phi$, the number of unsatisfied heavy constraints is upper bounded by the number of unsatisfied constraints in $\Psi$, which is at most $O(\sqrt{\delta} m)$.

Suppose $x_i x_j x_k = c_{ijk}$ is a light constraint. We claim that 
this constraint is satisfied by 
$\hat{x}_1,\dots \hat{x}_n$
if (1) it is satisfied by the ground truth solution $x^*_1,\dots x^*_n$ and (2) its representatives -- constraints
$x_i = \sigma_i$,
$x_j = \sigma_j$,
$x_k = \sigma_k$ -- are satisfied 
by both 
$x^*_1,\dots x^*_n$
and 
$\hat{x}_1,\dots \hat{x}_n$. Indeed, in this case, we have
$\hat{x}_i = \sigma_i = x^*_i$,
$\hat{x}_j = \sigma_j = x^*_j$,
$\hat{x}_k = \sigma_k = x^*_k$, and 
$\hat x_i \hat x_j \hat x_k = 
x^*_i x^*_j x^*_k = c_{ijk}$. Therefore, if $x_i x_j x_k = c_{ijk}$ is not satisfied by 
$\hat{x}_1,\dots \hat{x}_n$, then (1) it is also not satisfied in the ground truth solution; or (2) one of its representatives is not satisfied by $x^*_1,\dots x^*_n$ or $\hat{x}_1,\dots \hat{x}_n$.
Consequently, the number of light constraints  unsatisfied by $\hat{x}_1,\dots \hat{x}_n$ is upper bounded by the number of light  constraints unsatisfied by $x^*_1,\dots x^*_n$ plus the number of constraints in $\Psi$ not satisfied by either $x^*_1,\dots x^*_n$  or 
$\hat{x}_1,\dots \hat{x}_n$. The total number of such constraints is 
$O(\delta m) + O(\delta m) +O(\sqrt{\delta}m)= O(\sqrt{\delta}m)$. This concludes the proof of Theorem~\ref{thm:3LIN}.
\end{proof}
\subsection{Proof of Lemma~\ref{lem:soln-Psi}}\label{sec:lem:soln-Psi}
In this section, we prove Lemma~\ref{lem:soln-Psi}.
\begin{proof}[Proof of Lemma~\ref{lem:soln-Psi}]
We first upper bound the expected number of representatives for heavy constraints violated by the ground truth solution $x^*$. Let $R_{ij}$ be the set of representative constraints for constraints in $E_{ij}$. In other words, $R_{ij}$ is the set of constraints created by function \textsc{Create-H-Constraints} for set $E_{ij}$. Let $\Cost(E_{ij},x)$ and 
$\Cost(R_{ij}, x)$ be the numbers of constraints violated by solution $x$ in 
$E_{ij}$ and $R_{ij}$, respectively. We claim that the following bound holds.
\begin{lemma}\label{lem:ineq:heavy}
For every \emph{heavy} set $E_{ij}$,
\begin{equation}\label{ineq:heavy}
\E[\Cost(R_{ij},x^*)] \leq
8 \Cost(E_{ij},x^*) + 2e^{-\varepsilon^2 t/8} |E_{ij}|.   
\end{equation}
\end{lemma}
\begin{proof}
If $\Cost(E_{ij},x^*)\geq |E_{ij}|/4$, then $8 \Cost(E_{ij},x^*) \geq 2|E_{ij}|$. Since $\Cost(R_{ij},x^*)\leq |R_{ij}| = 2|E_{ij}|$, bound (\ref{ineq:heavy}) holds. For the rest of the proof, we shall assume that 
$\Cost(E_{ij},x^*) < |E_{ij}|/4$. We upper bound the probability that 
$\sigma_{ij}\neq x^*_i x^*_j$. To simplify the exposition, let us assume that $x^*_ix^*_j = 
1$ (the proof for the case $x^*_ix^*_j = -1$ is analogous). Write
$$
\Pr(\sigma_{ij} \neq x^*_ix^*_j)=\Pr(\sigma_{ij} \neq 1) = \Pr\Big(\sum_{k:(i,j,k)\in E_{ij}}c_{ijk}\tilde{x}_k \leq 0\Big).
$$
The expectation of $\tilde{x}_k$ is $\varepsilon x^*_k$. Thus, 
$\E[c_{ijk}\tilde{x}_k] = \varepsilon c_{ijk} x^*_k$. For each constraint $x_ix_jx_k=c_{ijk}$
satisfied by $x^*$, we have 
$c_{ijk} = x^*_ix^*_jx^*_k = x^*_k$ and 
$c_{ijk} x^*_k = (x^*_k)^2 = 1$.
Since the number of unsatisfied constraints is upper bounded by $|E_{ij}|/4$, we have 
$$
\E\Big[\sum_{k:(i,j,k)\in E_{ij}}
c_{ijk} \tilde{x}_k \Big]=
\sum_{k:(i,j,k)\in E_{ij}} 
\varepsilon c_{ijk} x^*_k
\geq 
\varepsilon |E_{ij}|/2.
$$
By Hoeffding's inequality,
$$
\Pr(\sigma_{ij}\neq 1) =
\Pr\Big(\sum_{k:(i,j,k)\in E_{ij}}
c_{ijk} \tilde{x}_k\leq 0\Big)\leq 
e^{-\frac{2(\varepsilon |E_{ij}|/2)^2}{4|E_{ij}|}}
=
e^{-\frac{\varepsilon^2 |E_{ij}|}{8}}.
$$
Recall that $E_{ij}$ is a heavy set and thus contains at least $t$ elements. Hence, 
$\Pr(\sigma_{ij}\neq x^*_i x^*_j)\leq e^{-\varepsilon^2t/8}$.

If $\sigma_{ij}\neq x^*_i x^*_j$, then we bound the number of unsatisfied constraints by $|R_{ij}| = 2|E_{ij}|$. This gives us the second term in bound~(\ref{ineq:heavy}). If 
$\sigma_{ij} = x^*_i x^*_j$, then all constraints $x_ix_j = \sigma_{ij}$ in $\Psi$ are satisfied by $x^*$. A constraint $x_k = \sigma_{ij}c_{ijk}$ in $\Psi$ is not satisfied by $x^*$ if and only if $x^*_i x^*_jx^*_k \neq c_{ijk}$. In other words, $x_k = \sigma_{ij}c_{ijk}$ is not satisfied only if constraint $x_i x_j x_k = c_{ijk}$ is not satisfied by $x^*$. The number of such constraints is $\Cost(|E_{ij}|,x^*)$. It is upper bounded by the first term in~(\ref{ineq:heavy})
\end{proof}

We now upper bound the expected number of representatives for light constraints violated by the ground truth solution $x^*$. Let $R_{i}$ be the set of representative constraints for constraints in $L_{i}$. In other words, $R_{i}$ is set of constraints created by function \textsc{Create-L-Constraints} for set $L_{i}$.
Let $\Cost(L_{i},x)$ and $\Cost(R_{i}, x)$ be the number of constraints violated by solution $x$ in $L_{i}$ and $R_{i}$, respectively. We prove the following lemma.
\begin{lemma}\label{lem:ineq:light}
For every set $L_{i}$,
\begin{equation}\label{ineq:light}
\E[\Cost(R_{i},x^*)] \leq 4 \Cost(L_{i},x^*) + 
\frac{6t}{\varepsilon^4}.
\end{equation}
\end{lemma}
\begin{proof}
If $\Cost(L_i,x^*)\geq |L_i|/4$, then the desired bound holds since, in this case, $\Cost(R_{i},x^*)\leq |R_i| = |L_i|\leq 4 \Cost(L_{i},x^*)$. So, we shall assume that 
$\Cost(L_i,x^*) < |L_i|/4$. 

Recall that $R_{i}$ contains $|L_i|$ identical constraints $x_i=\sigma_i$. These constraints are not satisfied by $x^*$ if $x^*_i \neq \sigma_i$. Similarly to the proof of Lemma~\ref{ineq:heavy}, we shall assume that $x^*_i =1$ (the proof for the case $x^*_i =-1$ is  analogous). We have
$$
\Pr(\sigma_{i} \neq x^*_i)=
\Pr(\sigma_{i} \neq 1) = \Pr\Big(\sum_{j,k:(i,j,k)\in L_{i}}
c_{ijk}\tilde{x}_j\tilde{x}_k
\leq 0\Big).
$$
The expected value of each term 
$c_{ijk}\tilde{x}_j\tilde{x}_k$ is 
$\varepsilon^2 c_{ijk}x^*_jx^*_k$ since $\tilde{x}_j$ and $\tilde{x}_k$ are independent random variables with means $\varepsilon x^*_j$ and $\varepsilon x^*_k$, respectively. Thus,
$$
\E\Big[\sum_{j,k:(i,j,k)\in L_{i}}
c_{ijk} \tilde{x}_j\tilde{x}_k \Big]=
\sum_{j,k:(i,j,k)\in E_{ij}} 
\varepsilon^2 c_{ijk} x^*_jx^*_k\geq 
\varepsilon^2 |L_{i}|/2.
$$
Here, we used that for every satisfied constraint $x_ix_jx_k=c_{ijk}$, we have 
$c_{ijk} x^*_jx^*_k = 1$ and the assumption that $\Cost(L_i,x^*)\leq |L_i|/4$. We now use McDiarmid's bounded difference inequality~\cite{mcdiarmid1989method}. Let $\rho_s$ be the number of occurrences of variable $\tilde{x}_s$ in the sum 
$\Sigma_i = \sum_{j,k:(i,j,k)\in L_{i}}
c_{ijk} \tilde{x}_j\tilde{x}_k$. If we change the value of variable $\tilde{x}_s$ from $-1$ to $1$ or from $1$ to $-1$, then sum $\Sigma_i$ can change by at most $2\rho_s$, since all coefficients $c_{ijk}\in \{\pm 1\}$. Consequently, $\sigma_i$ satisfies the bounded differences property with constants $\rho_s$.
By McDiarmid's inequality,
$$\Pr(\sigma_i \neq 1) = 
\Pr\Big(\sum_{j,k:(i,j,k)\in L_{i}}
c_{ijk} \tilde{x}_j\tilde{x}_k \leq 0\Big)
\leq e^{-\frac{2(\varepsilon^2 |L_i|/2)^2}{\sum_s (2\rho_s)^2}} = 
e^{-\frac{\varepsilon^4 |L_i|^2}{8\sum_s \rho_s^2}}.
$$
The total number of random variables participating in the sum $\Sigma_i$ counted with repetitions is $2|L_i|$. 
The number of times $x_j$ appears in the sum is at most $|E_{ij}|$, which is at most $t$ if $E_{ij}$ is light. On the other hand, if $E_{ij}$ is heavy, then so is every $(i,j,k)\in E(\Phi)$, and thus $x_j$ does not appear in the sum at all. In either case, $x_j$ appears in the sum at most $t$ times.
Thus, by H\"older's inequality, 
$\sum_s \rho_s^2\leq \|\rho\|_1\,\|\rho\|_{\infty} \leq 2|L_i|\cdot t$. Therefore, 
$\Pr(\sigma_i \neq 1)\leq 
e^{-\frac{\varepsilon^4 |L_i|}{16t}}$ and 
$$\E[\Cost(R_{i},x^*)] = \Pr(\sigma_i \neq 1) |L_i| \leq e^{-\frac{\varepsilon^4 |L_i|}{16t}} |L_i| \leq \frac{6t}{\varepsilon^4}.$$
Here, we used inequality $e^{-z}z\leq 1/e$ for $z = \frac{\varepsilon^4 |L_i|}{16t}$.
\end{proof}

To finish the proof of Lemma~\ref{lem:soln-Psi}, we sum the bounds obtained in Lemmas~\ref{lem:ineq:heavy} and~\ref{lem:ineq:light} over all heavy and light sets. The total number of violated constraints in $\Psi$ is upper bounded in expectation by
$$\underbrace{\sum_{i,j} 8\Cost(E_{ij},x^*) +
\sum_i 4\Cost(L_i,x^*)}_{%
\leq 24 \delta |E(\Phi)|} 
+ 
\underbrace{\sum_{i,j} 2e^{-\varepsilon^2 t/8} |E_{ij}|}_{%
6\delta |E(\Phi)|} 
+ 
\underbrace{\sum_i \frac{6t}{\varepsilon^4}}_{6t\varepsilon^{-4}n}.
$$
The first two terms together are upper bounded by the total number of violated constraints in $\Phi$ times 24, because each violated constraint can belong to at most three sets $E_{ij}$ or three sets $L_i$. The third term is upper bounded by $6\delta |E(\Phi)|$ because $t = 8\varepsilon^{-2}\ln \nicefrac{1}{\delta}$ and $\sum_{ij}|E_{ij}|\leq 3 |E(\Phi)|$. Finally, the last term is upper bounded by $6\varepsilon^{-4}t n= 
6(8\varepsilon^{-2}\ln \nicefrac{1}{\delta})\cdot \varepsilon^{-4}\cdot n < 48\,\delta |E(\Phi)|$, because 
$|E(\Phi)| \geq \ln \nicefrac{1}{\delta}\cdot \delta^{-1}\varepsilon^{-6}\cdot n$.
\end{proof}

\section{Hardness of Max \texorpdfstring{$k$}{k}-Lin with Oracle Advice}

In this section, we prove our hardness results for Max $k$-Lin problems in the advice model, i.e., Theorem \ref{thm:3lin-hardness} and Theorem \ref{thm:4lin-hardness}. We build towards these results in several steps; we briefly outline them below:
\begin{itemize}
\item Firstly, in Section \ref{sec:eth-3lin}, we introduce our first ingredient, Lemma \ref{lem:lc-hardness}, where we show that $(1/2 + \epsilon)$-approximation for Max $3$-Lin on $n$-variables requires $2^{c_\epsilon n}$-running time, assuming ETH. We show this by combining the $2^{c'_\epsilon n}$-running time lower bound for $(1 - \epsilon,\epsilon)$-Gap Label Cover (Lemma \ref{lem:GapLabelCover-Hardness}), with H\r{a}stad's reduction from Label Cover to Max $3$-Lin (Theorem \ref{thm:3-lin}). 
\item Next, in Section \ref{sec:eth-4lin}, we show that we can transfer the above hardness for Max $3$-Lin to Max $4$-Lin instances with up to linear average degree, using a simple combinatorial reduction (Corollary \ref{corr:4-lin}).
\item Then in Section \ref{sec:oracle-advice}, we give a generic tool (Lemma \ref{lem:dist}), which shows that $2^{c_\epsilon n}$-time lower bound for a problem can be transferred to give polynomial time hardness for the same problem in the variable-subset advice model with parameter $c''_\epsilon$.
\item Finally, in Section \ref{sec:hardness-proof}, we combine the above ingredients to prove Theorems \ref{thm:3lin-hardness} and \ref{thm:4lin-hardness}.
\end{itemize}

\subsection{ETH and Hardness of Max 3-Lin}                      \label{sec:eth-3lin}

We first review the Exponential Time Hypothesis (ETH) and Linear PCP Conjecture and prove the following hardness result for Max 3-Lin.

\begin{lemma}					\label{lem:lc-hardness}
	Assume the ETH and Linear Size PCP Conjecture (see Conjecture~\ref{conj:lin-pcp}).
	For some absolute constants $C', C'',C''' > 0$, $\varepsilon_0\in (0,1/2)$, and $\eta(\epsilon) = C'/\sqrt{\log(1/\epsilon)}$, the following holds. For every $\varepsilon \in (0, \varepsilon_0)$, there is no algorithm that given a Max $3$-Lin instance $\cI$ on $n$ variables and $2^{O(1/\epsilon)^{C'''}} n$ constraints, distinguishes between the following cases:
	\[
	{\bf Yes~Case}: {\sf Val}(\cI) \geq 1 - \eta(\epsilon)	
	\qquad\qquad\textnormal{and}\qquad\qquad
	{\bf No~Case}: {\sf Val}(\cI) \leq 1/2 + \eta(\epsilon).
	\]
	in time $2^{2^{-(1/\epsilon)^{C''}}n} \cdot \poly(n)$.
\end{lemma}

We point out that the above hardness result is folklore\footnote{See e.g., the TCS Stack Exchange post~\cite{rwurl} for a discussion on quantitatively similar bounds for Max $3$-SAT.}, and it is well-known that it can be derived by combining ETH, Linear Size PCP conjecture, and the techniques from \cite{Has01}. For the sake of completeness, we provide a proof of this lemma in this section. 


We remind the reader the Exponential Time Hypothesis (ETH) and definition of the Label Cover problem.
\begin{conjecture}[ETH~\cite{IPZ01}]				\label{conj:seth}
	There exists a constant $c \in (0,1)$ such that for all large enough integers $n$, the $3$-SAT problem on $n$ variables cannot be solved in time $2^{cn} \poly(n)$.
\end{conjecture}


\begin{definition}[Label Cover]
	An instance $\cL(U,V,E,\Sigma_U,\Sigma_V, \{\pi_e\}_{e \in E})$ of Label Cover is a $2$-CSP on a bipartite graph $(U,V,E)$. The left and right label sets of instance $\cL$ are $\Sigma_U$ and $\Sigma_V$, respectively. Every edge $(u,v) \in E$ is identified with a {\bf projection} constraint $\pi_{uv}:\Sigma_V \to \Sigma_U$. A labeling $\sigma:U \cup V \to \Sigma_U \cup \Sigma_V$ satisfies an edge $(u,v)$ if $\pi_{uv}(\sigma(v)) = \sigma(u)$. We denote the maximum fraction of constraints that can be satisfied by a labeling by ${\sf Val}(\cL)$.
	For every $0 < s < c \leq 1$, the objective of the $(c,s)$-Gap Label Cover problem on $\cL$ is to decide whether ${\sf Val}(\cL) \geq c$ or ${\sf Val}(\cL) \leq s$.
	
	We write $N = |U| + |V|$ and $K = |\Sigma_U| + |\Sigma_V|$ to denote the number of variables and labels in $\cL$.
	
\end{definition}

{\bf Linear Size PCPs}. We will use the following conjecture on linear size PCPs. 

\begin{conjecture}[Linear Size PCP Conjecture, Conjecture 4.2 in \cite{DinurGapETH}]					\label{conj:lin-pcp}
	For some $C_1, C_2 > 0$ and all sufficiently small $\epsilon > 0$, there exists a polynomial-time reduction from $3$-SAT to Label Cover that satisfies the following properties. Assume that the reduction maps a $3$-SAT instance $\Phi$ of size $m$ to a Label Cover instance $\cL = (U,V,E,\Sigma_U,\Sigma_V,\{\pi_e\}_{e \in E})$. Then,
	\begin{itemize}
		\item $|U|,|V| \leq (1/\epsilon)^{C_1}\cdot m$.
		\item $|\Sigma_U|,|\Sigma_V| \leq (1/\epsilon)^{C_2}$.
		\item If ${\sf Val}(\Phi) = 1$, then ${\sf Val}(\cL) = 1$.
		\item If ${\sf Val}(\Phi) < 1$, then ${\sf Val}(\cL) \leq \epsilon$.
	\end{itemize}
\end{conjecture}

The above conjecture posits that for any $\epsilon$, there exists a reduction from $3$-SAT on $m$-clauses to the $(1,\epsilon)$-Gap-Label-Cover problem, such that the number of variables is at most ${\rm poly}(1/\epsilon)$ times the number of clauses, and the label set size is also bounded by ${\rm poly}(1/\epsilon)$. It is instructive to compare this with the result of Moshkovitz--Raz~\cite{MR08}, which gives a reduction from $3$-SAT to $(1,\epsilon)$-label cover with a weaker guarantee of $O_\epsilon(\log m)$-blow up in the instance size. The Linear Size PCP Conjecture states that one can get (qualitatively) the same parameters in the reduction from $3$-SAT to $(1,\epsilon)$-Gap-Label Cover without the additional $O(\log m)$ blow up in the number of variables.

{\bf Degree Increment Lemma}. We shall also need the following lemma which says that one can increase the average degree of the Label Cover instance while preserving its optimal value. 

\begin{lemma}               \label{lem:deg-incr}
	The following holds for every integer $\ell \geq 1$. Let $\cL = (U,V,E,\Sigma_U, \Sigma_V, \{\pi_e\}_{e \in E})$ be a Label Cover instance. Then there exists a polynomial time procedure that constructs a Label Cover instance $\cL' = (U',V',E',\Sigma_U,\Sigma_V,\{\pi'_e\}_{e \in E'})$ which satisfies the following properties:
	\begin{itemize}
		\item ${\sf Val}(\cL') = {\sf Val}(\cL)$.
		\item $|U'| = \ell |U|$ and $|V'| = \ell |V|$.
		\item $|E'| = \ell^2 |E|$.
	\end{itemize}
\end{lemma}

The above lemma is folklore; we provide a proof of it in Appendix \ref{sec:deg-incr} for the sake of completeness. The following corollary follows immediately by combining Conjecture \ref{conj:lin-pcp} and Lemma \ref{lem:deg-incr}.

\begin{corollary}           \label{corr:lin-pcp}
	Assume Conjecture \ref{conj:lin-pcp}. For some $C_1, C_2 > 0$ and all sufficiently small $\epsilon > 0$, the following holds for every integer $\ell \geq 1$. There exists a polynomial-time reduction from $3$-SAT to Label Cover that satisfies the following properties. Assume that the reduction maps a $3$-SAT instance $\Phi$ of size $m$ to a Label Cover instance $\cL = (U,V,E,\Sigma_U,\Sigma_V,\{\pi_e\}_{e \in E})$. Then,
	\begin{itemize}
		\item $|U|,|V| \leq (1/\epsilon)^{C_1}\cdot \ell \cdot m$.
		\item $|\Sigma_U|,|\Sigma_V| \leq (1/\epsilon)^{C_2}$.
		\item If ${\sf Val}(\Phi) = 1$, then ${\sf Val}(\cL) = 1$.
		\item If ${\sf Val}(\Phi) < 1$, then ${\sf Val}(\cL) \leq \epsilon$.
		\item The average degree of $\cL$ is at least $\ell$.
	\end{itemize}
\end{corollary}

{\bf Sparsification Lemma}. Additionally, we will use the following sparsification lemma for $3$-SAT by Calabro, Impagliazzo, and Paturi~\cite{CIP06}.

\begin{lemma}[Sparsification Lemma~\cite{CIP06}]						\label{lem:sparse}
	For every $\gamma > 0$, there exists a deterministic algorithm that given a $3$-SAT formula $F$ on $n$-variables outputs a sequence of $3$-SAT formulas $F_1,\ldots,F_s$ such that 
	\begin{itemize}
		\item[1.] $s \leq 2^{\gamma n}$.
		\item[2.] $F$ is satisfiable if and only if at least one of $F_1,\ldots,F_s$ is satisfiable.
		\item[3.] Each formula $F_i$ is on $n$ variables. The number of clauses in each $F_i$ is at most $(1/\gamma)^9\cdot n$.
		\item[4.] The algorithm runs in time $2^{\gamma n}\cdot \poly(n)$, where the degree of the polynomial may depend on $\gamma$.
	\end{itemize}
\end{lemma}

\subsubsection{Reduction to Gap Label Cover}

In this section, we reduce $3$-SAT to Gap Label Cover.

\begin{lemma}
	\label{lem:GapLabelCover-Hardness}
	Assume ETH and the Linear Size PCP conjecture. For some constants $C > 2C_2 > 0$ and every choice of constants $0 < \epsilon  < c/2$, there exists no algorithm for $(1,\epsilon)$-Gap Label Cover that given an instance with at most $N$ vertices, $K = (1/\epsilon)^{C_2}$ labels, and at least $2^{5(1/\epsilon)^{C_2}} N$ constraints, decides whether the instance is completely satisfiable or at most $\epsilon$ satisfiable in time $2^{2^{-5(1/\epsilon)^C}\cdot NK} \poly(N)$. Here $c$ is the constant from Conjecture \ref{conj:seth}. 
\end{lemma}
\begin{proof}
	Let $C_1,C_2$ be the constants from Conjecture \ref{conj:lin-pcp} and let
	\[
	C:= C_1 + 2C_2 + 10.
	\] 
	We describe a Turing reduction from the $3$-SAT problem to the Gap Label Cover problem:
	
	{\bf Reduction}. Given a $3$-SAT instance $\Phi$ on $n$ variables, we do the following:
	
	\begin{itemize}
		\item Run the sparsification algorithm from Lemma \ref{lem:sparse} on $\Phi$ with $\gamma = \epsilon$ and get $3$-SAT instances $\Phi_1,\ldots,\Phi_s$. 
		\item For every $i \in [s]$, run the reduction from Corollary \ref{corr:lin-pcp} on $\Phi_i$ with $\ell = 2^{5(1/\epsilon)^{C_2}}$ and get a $(1,\epsilon)$-Gap Label Cover instance $\cL_i$.
		\item Solve each of the $(1,\epsilon)$-Gap Label Cover instances $\cL_i$ and output YES if at least one of the $\cL_i$ is a YES instance.
	\end{itemize}
	
	We make a couple of immediate observations. Note that for every $i \in [s]$, the formula $\Phi_i$ has $n$ variables and at most $m := (1/\gamma)^9\cdot n$ clauses. Furthermore, the corresponding label cover instance $\cL_i$ has at most $N = (1/\epsilon)^{C_1} \cdot \ell \cdot m = 2^{5(1/\epsilon)^{C_2}}(1/\epsilon)^{C_1 + 9} n$ variables, and $K \leq (1/\epsilon)^{C_2}$ labels. Furthermore, it has at least $2^{5K} N$ constraints.
	
	{\bf Completeness}. Suppose $\Phi$ is satisfiable. Then, for some $i \in [s]$, the corresponding $3$-SAT instance $\Phi_i$ is satisfiable, and thus the corresponding Label Cover instance $\cL_i$ is also satisfiable.
	
	{\bf Soundness}. Suppose $\Phi$ is not satisfiable, then for every $i \in [s]$, $\Phi_i$ is not satisfiable, and hence ${\sf Val}(\cL_i) < \epsilon$.
	
	Now suppose there exists a $2^{2^{-5(1/\epsilon)^{C}} \cdot NK}\poly(N)$ time algorithm for $(1,\varepsilon)$-Gap Label Cover. By running it on every instance $\cL_i$, we solve $3$-SAT in time:  
	\[
	2^{\epsilon n} \cdot 2^{2^{-5(1/\epsilon)^{C}} \cdot \left(2^{5(1/\epsilon)C_2}(1/\epsilon)^{C_1 + 9} n \cdot (1/\epsilon)^{C_2}\right)}\poly(N) \leq 2^{2\epsilon n} \poly(n),
	\]
	which refutes ETH if $\epsilon < c/2$.
\end{proof}

\subsubsection{Reduction from Gap Label Cover to Max 3-Lin}

Here, we recall H\r{a}stad's $3$-bit PCP-based reduction from Label Cover to Max $3$-Lin. 

\begin{theorem}[\cite{Has01}]					\label{thm:3-lin}
	There exists an increasing continuous function $\eta:[0,1] \to [0,1]$ with $\eta(0) = 0$ for which the following holds. Given an instance $\cL = (U,V,E,\Sigma_U,\Sigma_V,\{\pi_e\}_{e \in E})$ of $(1,\epsilon)$-Gap Label Cover with $n := |U| + |V|$ variables, and $k: = |\Sigma_U| + |\Sigma_V|$ labels, and average degree at least $2^{5k}$, there exists a randomized polynomial-time reduction that with probability at least $1 - 2^{-k/2}$, outputs an instance of Max $3$-Lin $\cI''$ such that
	\begin{itemize}
		\item If ${\sf Val}(\cL) = 1$, then ${\sf Val}(\cI'') \geq 1 - \eta(\epsilon)$.
		\item If ${\sf Val}(\cL) \leq \epsilon$, then ${\sf Val}(\cI'') \leq 1/2 + \eta(\epsilon)$.
        \item The number of variables in $\cI''$ is $n' = 2^{|\Sigma_U|}|U| + 2^{|\Sigma_V|}|V|$, and the number of constraints is at most $8n'/\epsilon^2$.
        \item All constraints in $\cI''$ are non-edge-weighted and distinct.
	\end{itemize} 
	Additionally, $\eta(\epsilon) = C'/\sqrt{\log(1/\epsilon)}$ for some constant $C'>0$.
\end{theorem}
\begin{proof}
    The reduction consists of two parts:
    \begin{itemize}
    \item First, we use H\r{a}stad's reduction~\cite{Has01} to reduce $(1,\epsilon)$-Gap Label Cover to  $(1 - \eta(\epsilon),1/2 + \eta(\epsilon))$-Gap Max $3$-Lin\footnote{For any $1\geq a>b \geq 1/2$, an instance of $(a,b)$-Gap Max $3$-Lin is the decision problem, where the input is a Max $3$-Lin instance, and the objective is to distinguish between the cases whether the optimal value of the instance is at least $a$ or at most $b$.}. Note that this will output a Max $3$-Lin instance $\cI'$ where the constraints are edge-weighted.
    \item Then we use a sub-sampling step to sample a non-edge-weighted instance $\cI''$ with no duplicate constraints.
    \end{itemize}

    To begin with, we recall the reduction from \cite{Has01} for the sake of completeness and observe the bounds on the number of variables and running time of the reduction. As is standard, the reduction is stated in the form of a dictatorship test, which we describe in Figure \ref{fig:pcp-red}: 
	
	\begin{figure}[ht!]
		\begin{mdframed}
			{\bf Input}. Let $\cL = (U,V,E,\Sigma_U,\Sigma_V,\{\pi_{e}\}_{e \in E})$ denote a $(1,\epsilon)$-Gap Label Cover instance as in the setting of the theorem. \\
			
			{\bf Long Code Tables}. For every $u \in U$ and $v \in V$, introduce long code table $f_u:\{0,1\}^{a} \to \{0,1\}$ and $f_v:\{0,1\}^b \to \{0,1\}$, where $a = |\Sigma_U|$ and $b = |\Sigma_V|$. Let $\mu = 2^{-100/\epsilon}$.\\
			
			{\bf Test}. The distribution over the $3$-Lin constraints is defined using the following process:
			
			\begin{enumerate}
				\item Sample an edge $(u,v) \sim E$ uniformly at random, and let $\pi_{uv}:[b] \to [a]$ be the constraint for $(u,v)$.
				\item Independently sample $\theta \in \{0,1\}$, $x \in \{0,1\}^{a}$, and $y \in \{0,1\}^{b}$ uniformly at random.
				\item Sample $z \in \{0,1\}^b$ as follows: for every $i \in [b]$, sample $z_i = 0$ with probability $1 - \mu$ and $z=1$ with probability $\mu$.
				\item Output the constraint
				\[
				f_u(x) \oplus f_v(y) \oplus f_v(y \oplus \pi_{uv}(x) \oplus z \oplus \theta \cdot {\bf 1}) = \theta.
				\]
			\end{enumerate}
			\footnotetext{Here, for a projection function $\pi:[b] \to [a]$ and a string $x \in \{0,1\}^a$, $\pi(x)$ denotes the string $x_{\pi(1)},x_{\pi(2)},\ldots,x_{\pi(b)}$}
		\end{mdframed}
		\caption{PCP Reduction to Max $3$-Lin}
		\label{fig:pcp-red}
	\end{figure}
	
	It is easy to see that the resulting Max $3$-Lin instance has $n' := 2^a |U| + 2^b |V| \leq 2^{k} \cdot n$ variables. The reduction runs in time $2^{O(k/\epsilon)}\poly(n)$. Furthermore, the analysis from H\r{a}stad's paper~\cite{Has01} shows that if $\cL$ is satisfiable, then the resulting instance has value at least $1 - \mu \geq 1 - \eta(\epsilon)$, and if $\cL$ has value at most $\epsilon$, then $\cI$ has value at most $1/2 +\eta(\epsilon)$.
	
	{\bf Removing Edge Weights}. Note that the instance (say, $\cI$) output by the above reduction is constraint-weighted; in particular, the weights define a probability distribution (say, $\nu$) over the set of constraints. Now consider the following new unweighted instance $\cI''$ that is constructed as follows:
	\begin{itemize}
		\item First, we construct the following intermediate instance $\cI'$. The variable set of $\cI'$ is the same as that of $\cI$. 
		Let $\eta = \eta(\epsilon)$. For $m := 8 n'/\eta^2$, sample constraints $e_1, e_2, \ldots, e_m \sim \nu$ independently with replacement, and include the constraints $e_1,\ldots,e_m$ in $\cI'$.
		\item For every constraint that has multiple copies in $\cI'$, delete all but one copy of the constraint. Call the resulting duplicate free instance $\cI''$. 
	\end{itemize}
	
	Clearly, by definition, the resulting instance still has $n' = 2^{a}|U| + 2^b|V|$ variables, and at most $m = 8n'/\eta^2$ constraints. It remains to be shown that the optimal value of $\cI''$ is nearly identical to $\cI$ with high probability. We show this using a couple of claims. We first show that the intermediate instance $\cI$ has nearly the same fraction of satisfied edges as $\cI'$ w.r.t. every labeling.
	
	\begin{claim}           \label{cl:cl-a}
		The following holds with probability at least $1 - e^{-O(n')}$ over the choice of $\cI'$. Suppose $m \geq 8n'/\eta^2$. Then for every labeling $\sigma$ of $\cI$, we have that 
		\[
		\left|\Pr_{e \sim \nu}\left(\sigma \mbox{~satisfies~} e \right) - \Pr_{e \sim \cI'}\left(\sigma \mbox{~satisfies~} e \right)\right|
		\leq \eta,
		\]
		where $\nu$ is the distribution over the constraints from the reduction in Figure \ref{fig:pcp-red}.
	\end{claim}
	\begin{proof}
		Fix a labeling $\sigma$ of $\cI$, and let $\alpha \in [0,1]$ denote the fraction of constraints satisfied by the labeling $\sigma$ in $\cI$. Now, for every $i \in [m]$, let $X_{\sigma,i}$ indicate whether the $i^{th}$ sampled constraint is satisfied by $\sigma$. Note that $\Ex[X_{\sigma,i}] = \alpha$. Then using Hoeffding's bound, we have that
		\[
		\Pr_{\cI'}\left(\left|\sum_{i \in [m]}X_{\sigma,i} - \alpha m \right| \geq \eta m \right) \leq e^{-\frac{\eta^2m}{4}} \leq e^{-2n'},  
		\]
		where the last inequality uses $m = 8n'/\eta^2$. The claim now follows by taking a union bound over all $2^{n'}$ possible labelings, and using our lower bound on $m$.
	\end{proof}
	Next, we show that the number of duplicate edges in $\cI'$ is not too large.
	  \begin{claim}               \label{cl:cl-b}
		Suppose $|E| \geq 2^{5k}n$. With probability at least $1 - 2^{-2k}$, the number of duplicate constraints in $\cI'$ is at most $\eta m$. 
	\end{claim}
	\begin{proof}
		Firstly, note that by definition of the distribution over the constraints $\nu$ (described in Figure \ref{fig:pcp-red}), we have that for any constraint $e$, $\nu(e):= \Pr_{e' \sim \nu}(e' = e) \leq 2^{-(a + b)}/|E|$. This is due to the observation that every constraint that can be generated by the process described in Figure \ref{fig:pcp-red} is specified by the choice of (i) the label cover edge $(u,v)$, (ii) the choice of the strings $x,y,z$ and the bit $\theta$. Furthermore, since $(u,v)$ is sampled uniformly from $E$ and independently, $x,y$ are sampled uniformly from $\{0,1\}^a$ and $\{0,1\}^b$ respectively, it follows that $\nu(e) \leq |E|^{-1}\cdot 2^{-(a + b)}$. 
		
		Therefore, for any $i,j \in [m]$, we have 
		\[
		\Pr_{\cI'}\Big(e_i = e_j\Big) = \sum_{e} \nu(e)^2 \leq \frac{2^{-(a + b)}}{|E|}\sum_{e} \nu(e) = \frac{2^{-k}}{|E|}.
		\]
		Therefore, for a fixed constraint $e_i$, $\cI'$ contains another copy of $e_i$ with probability at most $2^{-k} m/|E|$. Hence, the expected number of constraints that have another copy in $\cI'$ is at most 
		\[
		\frac{2^{-k}}{|E|}\cdot m^2 \leq \frac{2^{-k}}{|E|} \left(\frac{2^{2k} n}{\eta^2} \right)\cdot m \leq 2^{-2k} \eta m   
		\]
		where the first inequality uses the bound
        \[
            m = 8n'/\eta^2 \leq 8\cdot 2^{k} n/\eta^2 \leq 2^{2k}n/\eta^2,
        \]
         and the second inequality follows from the assumption that $|E| \geq 2^{5k}n$. The claim now follows using Markov's inequality.
	\end{proof}
	Therefore, by combining the guarantees of Claims \ref{cl:cl-a} and \ref{cl:cl-b}, we have that with probability at least $1 - 2^{-k}$, the following holds simultaneously for every labeling of $\cI$: if a labeling $\sigma$ satisfies $\alpha_\sigma$ fraction of constraints in $\cI$, then it satisfies $\alpha_{\sigma} \pm 2\eta$ fraction of constraints in $\cI''$. This concludes the proof of the lemma.
\end{proof}

\subsubsection{Proof of Lemma~\ref{lem:lc-hardness}}
Let $C > 2C_2 > 0$ be the constants from Lemma \ref{lem:GapLabelCover-Hardness}. Let $\cL$ be a $(1,\epsilon)$-Label Cover instance with $N$ variables and $K$ Labels as in the setting of Lemma \ref{lem:GapLabelCover-Hardness}. Note that by the guarantee of Lemma \ref{lem:GapLabelCover-Hardness}, we have $K \leq (1/\varepsilon)^{C_2}$, and the average degree of the instance is at least $2^{5K}$. Then we use the reduction from Theorem~\ref{thm:3-lin} to construct a $(1 - \eta(\epsilon),1/2 + \eta(\epsilon))$-Gap Max $3$-Lin instance $\cI$ with $n \leq 2^{2K}\cdot N $ variables, and at most $8n/\epsilon^2$ constraints. Furthermore, from the guarantee of Theorem \ref{thm:3-lin}, we know that $\cI$ is a YES instance if and only if $\cL$ is a YES instance. 

Now, suppose there exists a $2^{2^{-10(1/\epsilon)^{C}}n}\cdot \poly(n)$-time algorithm for solving $(1 - \eta(\epsilon),1/2 + \eta(\epsilon))$-Gap Max $3$-Lin.  Then there exists an algorithm for $(1,\epsilon)$-Gap Label Cover on $N$ variables, $K$ labels, and $2^{5K}N$ constraints, that runs in time
\[
2^{2^{-10(1/\epsilon)^{C}}n}\cdot\poly(n) \leq 
2^{2^{-10(1/\epsilon)^{C}} 2^{2(1/\epsilon)^{C_2}}N} \cdot2^{O(K)} \poly(N) 
\leq 2^{2^{-(1/\epsilon)^C} NK}\cdot \poly(N),
\]
where the first inequality uses the bound $K \leq (1/\epsilon)^{C_2}$, and the second inequality is due to the definition of $C \geq 2C_2$. Since this contradicts Lemma \ref{lem:GapLabelCover-Hardness}, it follows that there is no algorithm for $(1 - \eta(\epsilon), 1/2 + \eta(\epsilon))$-Gap Max $3$-Lin instances on $n$ variables with average degree at most $2^{O((1/\epsilon)^{C_2})}$.

\subsection{Hardness of Max 4-Lin Instances with Large Average Degree}              \label{sec:eth-4lin}

In this section, we reduce bounded degree Max $3$-Lin instances to Max $4$-Lin instances with large degrees.
It will be more convenient to work with $3$- and $4$-Lin constraints over the $0$-$1$ alphabet. Accordingly, the constraints will be of the form $x_i\oplus x_j \oplus x_k = b_{ijk}$ or $x_i\oplus x_j \oplus x_k \oplus x_l= b_{ijkl}$. We will say that the average degree of a 3-Lin or 4-Lin instance $\cI$ is the ratio between the number of constraints $|E|$ and variables $|V|$ in $\cI$, divided by the arity $r$ of the constraints, i.e., the average degree is defined as $|E|/(|V|r)$. 

\begin{lemma}						\label{lem:4-lin}
 For $0 < \epsilon,\delta \leq 1/2$, and a fixed integer parameter $t \in \mathbbm{N}$ there exists a polynomial-time deterministic reduction from  Max $3$-Lin to Max $4$-Lin that satisfies the following properties.
 \begin{itemize}
 \item If the input instance $\cI$ is $(1-\epsilon)$-satisfiable, then the output instance is also $(1-\epsilon)$-satisfiable.
 \item If the output instance $\cI'$ is $(1/2 + \delta)$-satisfiable, then the input instance $\cI$ is $(1/2 + \delta)$-satisfiable.
 \item The number of variables in $\cI'$ is $n + t$, where $n$ is the number of variables in $\cI$.
 \item The average degree of $\cI'$ is $ndt/(n + t)$.
 \end{itemize}
 \end{lemma}
\begin{proof}
Let 
 Let $\cI = (V = [n],E)$ be an instance of Max $3$-Lin. Given $\cI$, we construct an instance of Max $4$-Lin $\cI' = (V',E')$ as follows.
	
	{\bf Vertex Set}. Let $V' = V \cup V_1$, where $V_1 := \{y_1,\ldots,y_{t}\}$ is a set of $t$ new variables.
	
	{\bf Constraint Set}. For every constraint $e := x_i \oplus x_j \oplus x_k = b_{ijk}$ and every variable $y_r \in V_1$,  we add constraint
	\begin{equation}\label{constraint:I}
	e^r := x_i \oplus x_j \oplus x_k \oplus y_r = b_{ijk}.
	\end{equation}
	We will refer to the cloud of constraints corresponding to $e$ as $\cC_e$. Note that the degree of vertices in $V$ and $V_1$ are $dt$ and $nd$, respectively, which implies that the minimum degree of $\cI'$ is $\min\{dt, nd\}$.
	
	{\bf Analysis}. First, let us assume that ${\sf Val}(\cI) \geq 1 - \epsilon$. Denote an optimal labeling by $\sigma^*$. We extend $\sigma^*$ to a labeling $\sigma'$ of $V'$ by labeling all the variables in $V_1$ as $0$. The resulting labeling satisfies at least a $1-\varepsilon$ fraction of type-1 constraints and all type-2 constraints. Thus, it satisfies at least a $1 - \epsilon/2$ fraction of all the constraints in $\cI'$. Therefore, ${\sf Val}(\cI') \geq 1 - \epsilon/2$.   
	
	Now let us assume that $\cI'$ has an assignment, say $\sigma':V' \to \{0,1\}$, that satisfies at least a $(1/2 + \delta)$-fraction of the constraints in $\cI'$. Consider the following randomized labeling scheme. 
 
    \begin{itemize}
    \item Sample $r \sim [t]$ uniformly at random.
    \item Define labeling $\sigma:V \to \{0,1\}$ as $\sigma(x) = \sigma'(x) \oplus y_r$ for very $x \in V$.
    \end{itemize}
    Note that under the above randomized scheme, the expected fraction of constraints satisfied by $\sigma$ in $\cI$ is 
    \begin{align*}
    &\Ex_{\sigma}\Ex_{e = (i,j,k) \sim E}\Big[\mathbbm{1}\big(\sigma(i) \oplus \sigma(j) \oplus \sigma(k) = b_{ijk} \big)\Big] \\
    &= \Ex_{e = (i,j,k) \sim E}\Ex_{r \sim [t]}\Big[\mathbbm{1}\big(\sigma'(i) \oplus \sigma'(j) \oplus \sigma'(k) + y_r = b_{ijk} \big)\Big] \\
    &= \frac12 + \delta,
    \end{align*}
    which implies that for at least one choice of $y_r$, the corresponding labeling $\sigma$ of $V$ satisfies at least $(1/2 + \delta)$ fraction of variables in $\cI$, which concludes the soundness analysis.
\end{proof}

By combining Lemmas~\ref{lem:lc-hardness} and \ref{lem:4-lin}, we get the following corollary.

\begin{corollary}				\label{corr:4-lin}
Assume ETH and the Linear Size PCP conjecture.
    There exists $\delta_0 \in (0,1)$ such that the following holds. Then for every $\delta \in (0,\delta_0]$ there exists $d(\delta) \in \mathbbm{N}$ and $\epsilon(\delta) \in (0,1)$ such that given a $n$-variable Max $3$-Lin instance $\cI$ of average degree within $[d_0(\delta),n]$, there is no $2^{\epsilon(\delta)n}$ time algorithm that distinguish between the case when $\cI$ is at least $(1 - \delta)$ satisfiable and the case when $\cI$ is at most $1/2 + \delta$-satisfiable.
\end{corollary}

\begin{proof}
   Let $\eta$ and $\epsilon_0$ be as in Lemma \ref{lem:lc-hardness}. Let $\delta_0 = \eta^{-1}(\epsilon_0)$, and fix a $\delta \in (0,\delta_0)$. Let $\epsilon = \eta^{-1}(\delta)$. Note that since $\eta$ is increasing, we have that $\epsilon \in (0,\epsilon_0]$.
   Now, from Lemma \ref{lem:lc-hardness}. Then from Lemma~\ref{lem:lc-hardness} we know that there is no $2^{-c(\epsilon) n}$ time algorithm which can distinguish between the following cases for a Max-$3$-Lin instance on $n$-variables and $d(\epsilon) n$ constraints, where $d(\epsilon) = 2^{O(1/\epsilon)^{C'''}}$:
   \[
   {\bf YES~Case}: {\sf Val}(\cI) \geq 1 - \delta
   \qquad\qquad\textnormal{and}\qquad\qquad
   {\bf NO~Case}: {\sf Val}(\cI) \leq 1/2 + \delta,
   \]
   since $\delta = \eta(\epsilon)$. Now, we instantiate Lemma \ref{lem:4-lin} with the above instance and a fixed choice of $t \in \{1,2,\ldots,n\}$, which will yield a $4$-Lin instance $\cI'$ on $N := n + t$ variables and $ndt$ constraints with the following guarantees:
   \begin{itemize}
   \item If $\cI$ is a YES instance (as above), then ${\sf Val}(\cI') \geq 1 - \delta$.
   \item If $\cI$ is a NO instance (as above), then ${\sf Val}(\cI') \leq 1/2 + \delta$
   \end{itemize}
   Since there is no $2^{c(\epsilon(\delta))n}$-time algorithm that can distinguish between the YES and NO cases of $\cI$, since $N \in [n,2n]$, it follows that there is no $2^{c(\epsilon)N/2} = 2^{c(\eta^{-1}(\delta))N/2}$-time algorithm that can distinguish between the YES and NO cases of $\cI'$. Finally, also note that for any fixed choice of $t$, the average degree is $ndt/n+t = \Theta\left(d\cdot\min\{t,n\}\right)$, from which the claim follows.
\end{proof}

\subsection{Hardness with Oracle Advice}                \label{sec:oracle-advice}

In this section, we combine the hardness results from the previous section to prove ETH-based lower bounds for algorithms with oracle advice in the variable subset model. Our key tool here is the following lemma that shows that oracle advice can be simulated deterministically in near sub-exponential time. Let us say that an algorithm $\cal A$ is a $(c, s)$-approximation algorithm if given a $c$-satisfiable instance, it finds a solution that satisfies at least an $s$-fraction of the constraints.

\begin{lemma}\label{lem:dist}
    Suppose there exists a polynomial-time algorithm $\cA$ for Max $r$-Lin that given a $c$-satisfiable instance $\cI$ and advice with parameter $\epsilon$ in the variable-subset model, outputs a solution satisfying an $s$-fraction of the constraints with probability at least $0.9$ over the choice of the advice string. Then there exists a deterministic $(c,s)$-approximation algorithm $\cA'$ for Max $r$-Lin that runs in time  $2^{(\epsilon\log(4/\epsilon))n}\poly(n)$.
\end{lemma}
\begin{proof}
    Consider the  following algorithm: 
    
            {\bf Input}: Max $r$-Lin instance $\cI$.
            \begin{enumerate}
                \item For every $S \subseteq [n]$ of size at most $2\epsilon n$ and $\sigma_S \in \{0,1\}^S$, run algorithm $\cA$ with advice string $(S,\sigma_S)$ and let $x_{S,\sigma_S} \in \{0,1\}^n$ denote the corresponding labeling returned by algorithm.
                \item Output the best assignment among all the assignments $\{x_{S,\sigma_S}\}_{S,\sigma_S}$ computed in the above step.
            \end{enumerate}
    
    The running time of the algorithm is  at most 
    $$\sum_{t=0}^{2\epsilon n}\binom{n}{t} \cdot 2^{t}  \poly(n)\leq 
    \sum_{t=0}^{2\epsilon n}2^{H(t/n)n} \cdot 2^{t}  \poly(n) =2^{H(\epsilon)n + 2\epsilon n}\poly(n) = 2^{(\epsilon \log_2\frac{4}{\epsilon}) n} \poly(n).$$
    where $H(x) = x\log_2\frac{1}{x}$ is the Shannon entropy function.

    Let $x^*$ be the ground-truth solution. We assume that it satisfies at least a $c$-fraction of the constraints. Sample a random set $S$ by including every $i\in [n]$ with probability $\varepsilon$ (all decisions are independent). Then we are guaranteed that $\cal A$ finds a solution that satisfies an $s$ fraction of the constraints with probability at least $0.9$ given advice $(S, X^*|_S)$ (here $x^*|_S$ is the restriction of $x^*$ to $S$). The probability that $|S| \leq 2\varepsilon n$ is $1 - e^{\Omega(\varepsilon n)}$. Thus, for all sufficiently large $n$, with positive probability, we have that (1) $\cal A$, with advice $(S, X^*|_S)$,  finds a solution satisfying at least a $c$ fraction of the constraints and (2) $|S|\leq 2\varepsilon n$. Let $S_0$ be one of such sets $S$. When our algorithm goes over all $S$ of size at most $2\varepsilon n$, it also tries $S = S_0$. Then, for this $S=S_0$, it goes over assignments $\sigma_S$ including $x^*|_S$. For this choice of $(S, \sigma_S) = (S_0, x^*|_{S_0})$, $\cal A$ will find an solution satisfying at least a $c$-fraction of the constraints. As our algorithm returns the best of the solutions it finds at Step 1, it will output this or another solution satisfying at least a $c$-fraction of the constraints.
\end{proof}

\subsection{Proofs of Theorem \ref{thm:3lin-hardness} and Theorem \ref{thm:4lin-hardness}}          \label{sec:hardness-proof}

We now combine the ingredients from the previous section to prove the following hardness results:

\hardness*


\begin{proof}
    We apply Lemma~\ref{lem:lc-hardness} with $\varepsilon$ chosen so that $\delta = \eta(\varepsilon)$. We get that there is no algorithm that decides whether a 3-Lin instance is at most $\frac12 +\delta$ or at least $1-\delta$ satisfiable in time $2^{\left(2^{-(1/\varepsilon)^{C''}}\right)n} \poly(n)$.  Define $\varepsilon_0$ so that $\epsilon_0 \log_2 \frac{4}{\epsilon_0} \leq  2^{-(1/\varepsilon)^{C''}}$. Now the theorem statement follows from Lemma~\ref{lem:dist}.
\end{proof}

\ldeg*


\begin{proof}
    The proof of the theorem is almost identical to the proof of Theorem \ref{thm:3lin-hardness} -- we simply use Corollary~\ref{corr:4-lin} in place of Lemma~\ref{lem:lc-hardness}. Let $\delta_0$ be as in Corollary \ref{corr:4-lin}, and for a fixing of $\delta \in (0,\delta_0)$, let $\epsilon(\delta)$ be as in Lemma \ref{corr:4-lin}. Define $\varepsilon_0$ so that $\varepsilon_0\log_2\frac{4}{\varepsilon_0} \leq \epsilon(\delta)$. The desired result then follows from Lemmas~\ref{lem:lc-hardness} and~\ref{lem:dist}.
\end{proof}

\printbibliography

@inproceedings{GP19,
  title={Online algorithms for rent-or-buy with expert advice},
  author={Gollapudi, Sreenivas and Panigrahi, Debmalya},
  booktitle={Proceedings of the International Conference on Machine Learning},
  pages={2319--2327},
  year={2019},
  organization={PMLR}
}

@article{M18,
  title={A model for learned bloom filters and optimizing by sandwiching},
  author={Mitzenmacher, Michael},
  journal={Advances in Neural Information Processing Systems},
  volume={31},
  year={2018}
}

@inproceedings{BDSV18,
  title={Learning to branch},
  author={Balcan, Maria-Florina and Dick, Travis and Sandholm, Tuomas and Vitercik, Ellen},
  booktitle={International conference on machine learning},
  pages={344--353},
  year={2018},
  organization={PMLR}
}

@inproceedings{HIKV19,
  title={Learning-Based Frequency Estimation Algorithms.},
  author={Hsu, Chen-Yu and Indyk, Piotr and Katabi, Dina and Vakilian, Ali},
  booktitle={International Conference on Learning Representations},
  year={2019}
}

@article{LV21,
  title={Competitive caching with machine learned advice},
  author={Lykouris, Thodoris and Vassilvitskii, Sergei},
  journal={Journal of the ACM (JACM)},
  volume={68},
  number={4},
  pages={1--25},
  year={2021},
  publisher={ACM New York, NY}
}

@article{PSK18,
  title={Improving online algorithms via {ML} predictions},
  author={Purohit, Manish and Svitkina, Zoya and Kumar, Ravi},
  journal={Advances in Neural Information Processing Systems},
  volume={31},
  year={2018}
}

@article{GW95,
  title={Improved approximation algorithms for maximum cut and satisfiability problems using semidefinite programming},
  author={Goemans, Michel X and Williamson, David P},
  journal={Journal of the ACM (JACM)},
  volume={42},
  number={6},
  pages={1115--1145},
  year={1995},
  publisher={ACM New York, NY, USA}
}

@article{KKMO07,
  title={Optimal inapproximability results for {MAX-CUT} and other 2-variable {CSPs}?},
  author={Khot, Subhash and Kindler, Guy and Mossel, Elchanan and O’Donnell, Ryan},
  journal={SIAM Journal on Computing},
  volume={37},
  number={1},
  pages={319--357},
  year={2007},
  publisher={SIAM}
}

@inproceedings{MM17,
  title={Approximation algorithms for {CSPs}},
  author={Makarychev, Konstantin and Makarychev, Yury},
  booktitle={Dagstuhl Follow-Ups},
  volume={7},
  year={2017},
  organization={Schloss Dagstuhl-Leibniz-Zentrum fuer Informatik}
}

@inproceedings{CIP06,
	title={A duality between clause width and clause density for SAT},
	author={Calabro, Chris and Impagliazzo, Russell and Paturi, Ramamohan},
	booktitle={Conference on Computational Complexity},
	pages={7--pp},
	year={2006},	
}

@inproceedings{DinurGapETH,
	title={Mildly exponential reduction from gap-{3SAT} to polynomial-gap label-cover},
	author={Dinur, Irit},
	booktitle={Electronic colloquium on computational complexity ECCC; research reports, surveys and books in computational complexity},
	pages={128},
	year={2016}
}

@article{Has01,
	title={Some optimal inapproximability results},
	author={H{\aa}stad, Johan},
	journal={Journal of the ACM (JACM)},
	volume={48},
	number={4},
	pages={798--859},
	year={2001},
	publisher={ACM New York, NY, USA}
}

@article{IPZ01,
  title={Which problems have strongly exponential complexity?},
  author={Impagliazzo, Russell and Paturi, Ramamohan and Zane, Francis},
  journal={Journal of Computer and System Sciences},
  volume={63},
  number={4},
  pages={512--530},
  year={2001},
  publisher={Elsevier}
}

@article{MR08,
  title={Two-query {PCP} with subconstant error},
  author={Moshkovitz, Dana and Raz, Ran},
  journal={Journal of the ACM (JACM)},
  volume={57},
  number={5},
  pages={1--29},
  year={2008},
  publisher={ACM New York, NY, USA}
}

@article{mcdiarmid1989method,
  title={On the method of bounded differences},
  author={McDiarmid, Colin},
  journal={Surveys in combinatorics},
  volume={141},
  number={1},
  pages={148--188},
  year={1989},
  publisher={Norwich}
}

@misc{rwurl,
  title = {Super-polynomial time approximation algorithms for {MAX 3SAT}},
    author       = {TCS Stack Exchange},
    howpublished = {\url{https://cstheory.stackexchange.com/questions/9350/super-polynomial-time-approximation-algorithms-for-max-3sat}},
    note = {Accessed: 2024-04-02}
}

@article{DILMV21,
  title={Faster matchings via learned duals},
  author={Dinitz, Michael and Im, Sungjin and Lavastida, Thomas and Moseley, Benjamin and Vassilvitskii, Sergei},
  journal={Advances in neural information processing systems},
  volume={34},
  pages={10393--10406},
  year={2021}
}

@inproceedings{lu2021generalized,
  title={Generalized Sorting with Predictions},
  author={Lu, Pinyan and Ren, Xuandi and Sun, Enze and Zhang, Yubo},
  booktitle={Proceedings of the Symposium on Simplicity in Algorithms},
  pages={111--117},
  year={2021},
}

@inproceedings{carvalho2023learnedsort,
  title={LearnedSort as a learning-augmented SampleSort: Analysis and Parallelization},
  author={Carvalho, Ivan and Lawrence, Ramon},
  booktitle={Proceedings of the International Conference on Scientific and Statistical Database Management},
  pages={1--9},
  year={2023}
}

@article{bai2024sorting,
  title={Sorting with predictions},
  author={Bai, Xingjian and Coester, Christian},
  journal={Advances in Neural Information Processing Systems},
  volume={36},
  year={2024}
}

@inproceedings{kraska2018case,
  title={The case for learned index structures},
  author={Kraska, Tim and Beutel, Alex and Chi, Ed H and Dean, Jeffrey and Polyzotis, Neoklis},
  booktitle={Proceedings of the International Conference on Management of Data},
  pages={489--504},
  year={2018}
}

@article{munoz2017revenue,
  title={Revenue optimization with approximate bid predictions},
  author={Munoz, Andres and Vassilvitskii, Sergei},
  journal={Advances in Neural Information Processing Systems},
  volume={30},
  year={2017}
}

@inproceedings{mahdian2007allocating,
  title={Allocating online advertisement space with unreliable estimates},
  author={Mahdian, Mohammad and Nazerzadeh, Hamid and Saberi, Amin},
  booktitle={Proceedings of the Conference on Electronic Commerce},
  pages={288--294},
  year={2007}
}

@inproceedings{vee2010optimal,
  title={Optimal online assignment with forecasts},
  author={Vee, Erik and Vassilvitskii, Sergei and Shanmugasundaram, Jayavel},
  booktitle={Proceedings of the Conference on Electronic Commerce},
  pages={109--118},
  year={2010}
}

@inproceedings{devanur2009adwords,
  title={The adwords problem: online keyword matching with budgeted bidders under random permutations},
  author={Devanur, Nikhil R and Hayes, Thomas P},
  booktitle={Proceedings of the Conference on Electronic Commerce},
  pages={71--78},
  year={2009}
}

@inproceedings{lattanzi2020online,
  title={Online scheduling via learned weights},
  author={Lattanzi, Silvio and Lavastida, Thomas and Moseley, Benjamin and Vassilvitskii, Sergei},
  booktitle={Proceedings of the Symposium on Discrete Algorithms},
  pages={1859--1877},
  year={2020},
}

@inproceedings{mitzenmacher2020scheduling,
  title={Scheduling with Predictions and the Price of Misprediction},
  author={Mitzenmacher, Michael},
  booktitle={Proceedings of the Innovations in Theoretical Computer Science Conference},
  year={2020},
}

@article{CdGLP24,
  title={Max-Cut with $\epsilon$-Accurate Predictions},
  author={Cohen-Addad, Vincent and d'Orsi, Tommaso and Gupta, Anupam and Lee, Euiwoong and Panigrahi, Debmalya},
  journal={arXiv preprint arXiv:2402.18263},
  year={2024}
}

@article{FKL02,
  title={Improved approximation of {Max-Cut} on graphs of bounded degree},
  author={Feige, Uriel and Karpinski, Marek and Langberg, Michael},
  journal={Journal of Algorithms},
  volume={43},
  number={2},
  pages={201--219},
  year={2002},
  publisher={Elsevier}
}

@article{HK22,
  title={Approximating max-cut on bounded degree graphs: Tighter analysis of the {FKL} algorithm},
  author={Hsieh, Jun-Ting and Kothari, Pravesh K},
  journal={arXiv preprint arXiv:2206.09204},
  year={2022}
}

@book{Karp,
  title={Reducibility among combinatorial problems},
  author={Karp, Richard M},
  year={2010},
  publisher={Springer}
}

@article{BEX24,
  title={Parsimonious Learning-Augmented Approximations for Dense Instances of NP-hard Problems},
  author={Bampis, Evripidis and Escoffier, Bruno and Xefteris, Michalis},
  journal={arXiv preprint arXiv:2402.02062},
  year={2024}
}

\appendix
\section{Degree Increment Lemma}          \label{sec:deg-incr}

We prove Lemma \ref{lem:deg-incr} here.

\begin{proof}
    Given $\cL$, we construct the label cover instance $\cL'(U',V',E',\Sigma_U,\Sigma_V,\{\pi'_e\}_{e \in E'})$ as follows.

    {\bf Vertex Set}. For every $u \in U$, we introduce $\ell$ copies of the vertex $u$, which we call $(u,1),\ldots,(u,\ell)$. Similarly, for every vertex $v \in V$, we introduce $\ell$ copies of the vertex $v$, which we call $(v,1),\ldots,(v,\ell)$. Let $U' = \{(u,i) : u \in U, i \in [\ell]\}$, and similarly, let $V' = \{(v,i) : v \in V, i \in [\ell]\}$. 

    {\bf Label Set}. The left and right label sets are still $\Sigma_U$ and $\Sigma_V$.

    {\bf Constraint Set}. For every constraint $e = (u,v) \in E$, and for every pair of indices $(i,j) \in [\ell] \times [\ell]$, we introduce an edge $e_{ij} = \{(u,i),(v,j)\}$ in $E'$, and define the corresponding projection function $\pi'_{e_{ij}} = \pi_e$. 
    
    Note that by construction $|U'| = \ell |U|$, $|V'| = \ell|V|$, and $|E'| = \ell^2|E|$. This concludes the description of the reduction. All that remains is to claim that ${\sf Val}(\cL) = {\sf Val}(\cL')$. To that end, let us first show that ${\sf Val}(\cL') \geq {\sf Val}(\cL)$. Let $\sigma$ be an optimal labeling of $\cL$. Given $\sigma$, we construct a labeling $\sigma'$ of $\cL'$ as follows. For every $(u,i) \in U'$, we let $\sigma'(u,i) = \sigma(u)$, and similarly, for every $(v,j) \in V'$, we let $\sigma'(v,j) = \sigma(v)$. Then the fraction of constraints in $\cL'$ satisfied by $\sigma'$ is
    \begin{align*}
        \Pr_{(u',v') \sim E'} \left[ \pi'_{(u',v')}(\sigma'(v')) = \sigma'(u') \right] 
        &= \Ex_{i,j \sim [\ell] \times [\ell]} \Pr_{(u,v) \sim E} \left[ \pi'_{(u,i),(v,j)}(\sigma'(u,i)) = \sigma'(v,j)\right] \\
        &= \Ex_{i,j \sim [\ell] \times [\ell]} \Pr_{(u,v) \sim E} \left[ \pi'_{(u,v)}(\sigma(u)) = \sigma(v)\right] \\
        &= {\sf Val}(\cL),
    \end{align*}
    which shows that there exists a labeling of $\cL'$ which can satisfy at least ${\sf Val}(\cL)$ fraction of constraints in $\cL'$. This establishes that ${\sf Val}(\cL') \geq {\sf Val}(\cL)$. 
    
    Conversely, let $\sigma'$ be an optimal labeling of ${\sf Val}(\cL)$. Consider the following randomized labeling procedure for~$\cL$. 
    \begin{itemize}
    \item Sample $i \sim [\ell]$ and $j \sim [\ell]$ uniformly at random.
    \item Assign $\sigma(u) = \sigma'(u,i)$ and $\sigma(v) = \sigma'(v,j)$ for every $u \in U$ and $v \in V$ respectively.
    \end{itemize}

    Then, randomizing over the choice of $\sigma$, the expected fraction of constraints that are satisfied by the labeling $\sigma$ is 
    \begin{align*}
        \Ex_{\sigma}\Pr_{(u,v) \sim E}\left[\pi_{(u,v)}(\sigma(v)) = \sigma(u) \right] 
        &= \Ex_{i,j \sim [\ell] \times [\ell]} \Pr_{(u,v) \sim E} \left[ \pi_{(u,v)}(\sigma'(u,i)) = \sigma'(v,j)\right] \\
        &= \Ex_{i,j \sim [\ell] \times [\ell]} \Pr_{(u,v) \sim E} \left[ \pi'_{(u,i),(v,j)}(\sigma'(u,i)) = \sigma'(v,j)\right] \\
        &= {\sf Val}(\cL'),
    \end{align*}
    which implies that there exists a labeling of $\sigma$ of $\cL$ which satisfies at least ${\sf Val}(\cL')$ fraction of constraints, thus implying that ${\sf Val}(\cL) \geq {\sf Val}(\cL')$.
\end{proof}

\end{document}